\documentclass[conference]{IEEEtran}

\usepackage{wrapfig,epsfig}
\usepackage{graphicx, subfigure}
\usepackage{epstopdf}
\usepackage[linesnumbered,ruled,vlined]{algorithm2e}
\usepackage{amsmath}
\usepackage{amssymb}
\usepackage{url}
\usepackage{multirow}
\usepackage[english]{babel}
\usepackage{times}

\newtheorem{definition}{Definition}
\newtheorem{theorem}{Theorem}
\newtheorem{lemma}{Lemma}

\newtheorem{example}{Example}

\title{Efficient Processing of Reachability and Time-Based Path Queries in a Temporal Graph}

\author{
\IEEEauthorblockN{Huanhuan Wu{\small $^{*}$},  Yuzhen Huang{\small $^{*}$}, James Cheng{\small $^{*}$}, Jinfeng Li{\small $^{*}$}, Yiping Ke{\small $^{\#}$}}
\IEEEauthorblockA{{\small $^*$}Department of Computer Science and Engineering, The Chinese University of Hong Kong\\
Emails: \{hhwu,yzhuang,jcheng,jfli\}@cse.cuhk.edu.hk}
$^{\#}$School of Computer Engineering, Nanyang Technological University\\
Email: ypke@ntu.edu.sg
}

\begin{document}

\maketitle

\begin{abstract}

A temporal graph is a graph in which vertices communicate with each other at specific time, e.g., $A$ calls $B$ at 11 a.m. and talks for 7 minutes, which is modeled by an edge from $A$ to $B$ with starting time ``11 a.m.'' and duration ``7 mins''. Temporal graphs can be used to model many networks with time-related activities, but efficient algorithms for analyzing temporal graphs are severely inadequate. We study fundamental problems such as answering reachability and time-based path queries in a temporal graph, and propose an efficient indexing technique specifically designed for processing these queries in a temporal graph. Our results show that our method is efficient and scalable in both index construction and query processing.

\end{abstract}

\section{Introduction}  \label{sec:intro}


Graph has been extensively used to model and study the structures of various online social networks, mobile communication networks, e-commerce networks, email networks, etc. In these graphs, vertices are users or companies, while edges model the relationship between them. However, there is one type of information that is often missing in these graphs for simplicity of analysis: in reality, a relationship occurs at a specific time and lasts for a certain period. Formally, this can be modeled as a \textbf{temporal graph}, in which each edge is represented by $(u,v,t,\lambda)$, indicating that the relationship from $u$ to $v$ starts at time $t$ and lasts for a duration of $\lambda$. There may be multiple edges between $u$ and $v$ indicating their relationship occurring in different time periods.

Temporal graphs can be used to model and study many time-related activities in the above-mentioned graphs. For example, users follow or tag other users in different periods in online social networks; friends chat with each other in different time periods in mobile phone networks; people send messages to each other at different times in email networks or instant messaging networks; customers buy products from sellers at different times in online shopping platforms, or different types of transactions happened between different parties in different periods in e-commerce networks, etc.

Research on \textbf{non-temporal graphs} (i.e., general graphs without time information) has been extensively studied. However, for temporal graphs, even some fundamental problems have not been well studied (e.g., graph traversal, connected components, reachability, ``shortest paths'', etc.). In this paper, we study the problems of computing the reachability and the ``shortest path distance'' from a vertex to another vertex in a temporal graph.

Graph reachability and shortest path both have numerous important applications. However, in a temporal graph, the problems become more complicated due to the order imposed by time. For example, consider a toy train-schedule graph shown in Figure~\ref{fig:transform}(a), where the number next to each edge is the day  (e.g., Day 1, Day 2, etc.) that a train departs, and assume that the duration of each train takes 2 days. Suppose now one wants to travel from {\tt a} to {\tt d}. If he chooses to go to {\tt d} via {\tt b}, then he can leave either on Day 1 or Day 2, and he will reach {\tt d} on Day 6. Now suppose that he wants to depart later, on Day 4, then he cannot reach {\tt d} because the train departing on Day 4 from {\tt a} reaches {\tt c} on Day 6, but the train leaves from {\tt c} on Day 5 to {\tt d}. However, if we do not consider the time information, then the traveler may still take the train on Day 4 from {\tt a} to {\tt c}, not realizing that he would not catch the train from {\tt c} to {\tt d} in this case.

Given two vertices, $u$ and $v$, in a temporal graph, and a time interval $T$, we study how to compute (1)whether $u$ can reach $v$ within $T$, (2)the \emph{earliest time} $u$ can \emph{reach} $v$ within $T$, and (3)the \emph{duration of a fastest path} from $u$ to $v$ within $T$. The \emph{reachability}, \emph{earliest-arrival time}, and \emph{minimum duration} from source vertices to target vertices have been found useful in the study of temporal networks such as temporal graph connectivity~\cite{KempeKK02jcss}, temporal betweenness and closeness~\cite{PanS11physrev,SantoroQFCA11corr}, temporal connected components in~\cite{TangMML10ccr}, information propagation~\cite{ClementiP10}, information latency~\cite{CasteigtsFMS11ipps,KossinetsKW08kdd}, temporal efficiency and clustering coefficient~\cite{TangMML10ccr}, temporal small-world behavior study~\cite{TangSMML10phyrev}, etc.


The above cited works, however, did not focus on the design of efficient algorithms to compute reachability, earliest-arrival time and minimum duration, and their results were mostly obtained from small temporal graphs. Wu et al.~\cite{WuCHKLX14pvldb} made a significant improvement over existing algorithms~\cite{XuanFJ03ijfcs} and their algorithms can handle much larger temporal graphs than existing works. However, their algorithms were not designed for online querying, while in many applications it is demanding to find the reachability,  earliest-arrival time or minimum duration from a source vertex to a target vertex in real time. Wang et al.~\cite{WangLYXZ15sigmod} presented an indexing method to answer online queries of earliest-arrival time or minimum duration from a source vertex to a target vertex. However, their indexing method cannot scale to large temporal graphs. In addition, their indexing method does not support dynamic update, which is practically important for temporal graphs since updates are frequent in most real-world temporal graphs. In view of this, we propose an index to support efficient online querying for large temporal graphs, which also supports efficient dynamic update. 

Our method first transforms a temporal graph into a new graph which is a \emph{directed acyclic graph} (\emph{DAG}), on which existing indexing methods for reachability querying~\cite{ChengHWF13sigmod,JinW13pvldb,SeufertABW13icde,TrisslL07sigmod,SchaikM11sigmod,WeiYLJ14pvldb,YanoAIY13cikm,YildirimCZ12vldb,ZhuLWX14sigmod} can also be applied. However, this DAG is often significantly larger than the DAGs that are handled by existing methods. It also possesses of unique properties of temporal graphs, while all existing methods were designed for handling non-temporal graphs. Thus, more scalable methods that also consider the properties of temporal graphs need to be designed.

We propose \textbf{TopChain}, which is a labeling scheme for answering reachability queries. A labeling scheme, e.g., 2-hop label~\cite{CohenHKZ02soda}, constructs two labels for each vertex $v$, $L_{in}(v)$ and $L_{out}(v)$, where $L_{in}(v)$ and $L_{out}(v)$ are the set of vertices that can reach $v$ and that are reachable from $v$, respectively. A query whether $u$ can reach $v$ is answered by intersecting $L_{out}(u)$ and $L_{in}(v)$, since there exists a common vertex in $L_{out}(u)$ and $L_{in}(v)$ if $u$ can reach $v$. However, $L_{in}(v)$ and $L_{out}(v)$ are often too large, and various methods have been proposed to reduce their sizes~\cite{ChengHWF13sigmod,JinW13pvldb,WangLYXZ15sigmod,YanoAIY13cikm,ZhuLWX14sigmod}. 

TopChain decomposes an input DAG into a set of chains~\cite{Simon88tcs,Jagadish90tods}, where a chain is an ordered sequence of vertices such that each vertex can reach the next vertex in the chain. Thus, $L_{in}(v)$ and $L_{out}(v)$ only need to keep the last and first vertex in a chain that can reach $v$ and that is reachable from $v$, respectively. However, the number of chains can still be too large for a large graph, and as a solution, TopChain ranks the chains and only uses the top $k$ chains for each vertex. In this way, the size of the labels is kept to at most $2k$ for each vertex, and index construction takes only linear time, as $k$ is a small constant. The $k$ labels may not be able to answer every query, and thus online search may still be required. However, the labels can be employed to do effective pruning and online search converges quickly.




The contributions of our work are summarized as follows:


\begin{itemize}
  \item We propose an efficient indexing method, TopChain, for answering reachability and time-based path queries in a temporal graph, which is useful for analyzing real-world networks with time-based activities~\cite{WuCHKLX14pvldb}.
  \item TopChain has a linear index construction time and linear index size. Although existing methods can be applied to our transformed graph for answering reachability queries, our method is the only one that makes use of the properties of a temporal graph to design the indexing scheme. 
      TopChain also applies the properties of temporal graphs to devise an efficient algorithm for dynamic update of the index.

  \item We evaluated the performance of TopChain on a set of 15 real temporal graphs. Compared with the state-of-the-art reachability indexes~\cite{SeufertABW13icde,SchaikM11sigmod,WeiYLJ14pvldb,YildirimCZ12vldb,ZhuLWX14sigmod}, TopChain is from a few times to a few orders of magnitude faster in query processing, with a smaller or comparable index size and index construction cost. Compared with TTL~\cite{WangLYXZ15sigmod}, TopChain has significantly better indexing performance, is faster in query processing for most of the datasets, is more scalable, and supports efficient update maintenance.

\end{itemize}


\noindent \textbf{Paper outline. } Section~\ref{sec:def} defines the problem. Section~\ref{sec:transform} describes graph transformation. Sections~\ref{sec:index} and~\ref{sec:query} present the details of indexing and query processing. Section~\ref{sec:improvement} presents some improvements on labeling. Section~\ref{sec:result} reports experimental results. Section~\ref{sec:related} discusses related work and Section~\ref{sec:conclude} gives the concluding remarks.

\section{Problem Definition}  \label{sec:def}

\begin{table}[!tbp]
\caption{Frequently-used notations} \label{tab:notations}
\begin{center}
\resizebox{\linewidth}{!}{
\begin{tabular}{|c|l|}
  \hline
  Notation & Description \\
  \hline
  $\mathcal{G} = (\mathcal{V}, \mathcal{E})$ & A temporal graph \\
  \hline
  $e = (u, v, t, \lambda) \in \mathcal{E}$ & A temporal edge\\
  \hline
  $t(e)$ & The starting time of edge $e$ \\
  \hline
  $\lambda(e)$ & The traversal time of edge $e$ \\
  \hline
  $\Pi(u,v)$ & The set of temporal edges from $u$ to $v$ \\
  \hline
  $\pi(u,v)$ & The number of temporal edges from $u$ to $v$ \\
  \hline
  $\pi$ & Max. \# of temporal edges between any two vertices in $\mathcal{G}$ \\
  \hline
  $\Gamma_{out}(u,\mathcal{G})$ ($\Gamma_{in}(u,\mathcal{G})$)  & The set of out-neighbors (in-neighbors)  of   a vertex $u$ in $\mathcal{G}$ \\
  \hline
  $d_{out}(u,\mathcal{G}) (d_{in}(u,\mathcal{G}))$ & The out-degree (in-degree) of a vertex $u$ in $\mathcal{G}$ \\
  \hline
  $G=(V, E)$ & A directed acyclic graph (DAG) \\
  \hline
  $\mathbf{T}_{out}(v)$  & The set of distinct starting (arrival) time \\ ($\mathbf{T}_{in}(v)$) & of out-edges (in-edges) of $v$ \\
  \hline
  $V_{out}(v)$  ($V_{in}(v)$) & The set of vertices, $\langle v,t \rangle$, for each $t \in \mathbf{T}_{out}(v)$ ($\mathbf{T}_{in}(v)$)\\
  \hline
  $L_{out}(v)$ ($L_{in}(v)$) & The set of out-labels (in-labels) of a vertex $v$ \\
  \hline
  $\mathbb{C}=\{C_1, \ldots, C_l\}$ & A chain cover of $G$\\
  \hline
  $code(v)$ & Chain code of a vertex $v$ \\
  \hline
  $RC(v)$ ($RC^{-1}(v)$) & The set of chains that $v$ can reach (can reach $v$)\\
  \hline
  $first_v(C)$  ($last_v(C)$)  & The first (last) vertex in $C$ that $v$ can reach (can reach $v$) \\
  \hline
  $RF(v)$ ($RL(v)$)  & Set of first (last) reachable vertices in $RC(v)$ ($RC^{-1}(v)$) \\
  \hline
  $RFcode(v)$ ($RLcode(v)$)  & Set of chain codes of vertices in $RF(v)$ ($RL(v)$)\\
  \hline
  $top_{k}(RFcode(v))$  & The first $k$ chain codes in $RFcode(v)$ \\ ($top_{k}(RLcode(v))$) &($RLcode(v)$) \\
  \hline
\end{tabular}
}
\end{center}
\end{table}

\begin{figure}[!tbp]
\begin{center}
\includegraphics[width = 3.2in]{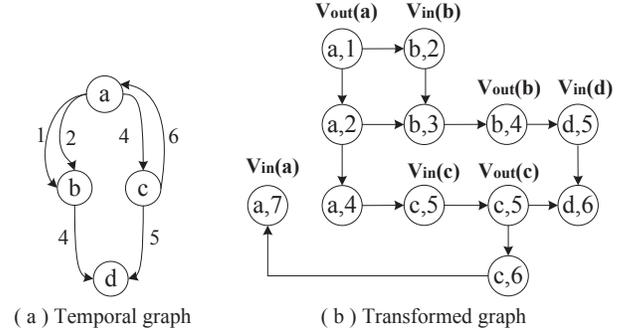}
\caption{A temporal graph $\mathcal{G}$ and its transformed graph $G$} \label{fig:transform}
\end{center}
\end{figure}

Let $\mathcal{G} = (\mathcal{V}, \mathcal{E})$ be a temporal graph, where $\mathcal{V}$ is the set of vertices in $\mathcal{G}$ and $\mathcal{E}$ is the set of edges in $\mathcal{G}$. An edge $e \in \mathcal{E}$ is a quadruple $(u, v, t, \lambda)$, where $u, v \in \mathcal{V}$, $t$ is the \textbf{starting time}, and $\lambda$ is the \textbf{traversal time} to go from $u$ to $v$ starting at time $t$. We denote the starting time of $e$ by $t(e)$ and the traversal time of $e$ by $\lambda(e)$. Alternatively, we can consider that $e$ is \textbf{active} during the period $[t, t+\lambda]$.

If edges are undirected, then the starting time and traversal time of an edge are the same from $u$ to $v$ as from $v$ to $u$. We focus on directed temporal graphs in this paper since an undirected edge can be modeled by two bi-directed edges.

We denote the set of temporal edges from $u$ to $v$ in $\mathcal{G}$ by $\Pi(u,v)$, and the number of temporal edges from $u$ to $v$ in $\mathcal{G}$ by $\pi(u,v)$, i.e., $\pi(u,v)=|\Pi(u,v)|$. We also define the maximum number of temporal edges from $u$ to $v$, for any $u$ and $v$ in $\mathcal{G}$, by $\pi = \max\{\pi(u,v): (u,v) \in (\mathcal{V} \times \mathcal{V})\}$. 


We define the set of \textbf{out-neighbors} of a vertex $u$ in $\mathcal{G}$ as $\Gamma_{out}(u,\mathcal{G})$ $=\{v: (u, v, t, \lambda) \in \mathcal{E}\}$, and the \textbf{out-degree} of $u$ in $\mathcal{G}$ as $d_{out}(u,\mathcal{G})$ $=\sum_{v \in \Gamma_{out}(u,\mathcal{G})} \pi(u,v)$. Similarly, we define the \textbf{in-neighbors} and \textbf{in-degree} of $u$ as $\Gamma_{in}(u,\mathcal{G})=\{v: (v, u, t, \lambda) \in \mathcal{E}\}$ and $d_{in}(u,\mathcal{G})=\sum_{v \in \Gamma_{in}(u,\mathcal{G})} \pi(v,u)$.


A \textbf{temporal path} $P$ in a temporal graph $\mathcal{G}$ is a sequence of edges $P=\langle e_1, e_2, \ldots, e_p \rangle$, such that $e_i=(v_i, v_{i+1}, t_i, \lambda_i)$ $\in \mathcal{E}$ is the $i$-th temporal edge on $P$ for $1 \le i \le p$, and $(t_i + \lambda_i) \le t_{i+1}$ for $1 \le i < p$. Note that for the last edge $(v_{p}, v_{p+1}, t_{p}, \lambda_{p})$ on $P$, we do not put a constraint on $(t_{p} + \lambda_{p})$ since  $t_{p+1}$ is not defined for the path $P$. In fact, $(t_{p} + \lambda_{p})$ is the \textbf{ending time} of $P$, denoted by $end(P)$. We also define the \textbf{starting time} of $P$ as $start(P)=t_1$.  We define the \textbf{duration} of $P$ as $dura(P)=end(P)-start(P)$.


Based on the temporal paths, we give the definitions of minimum temporal paths~\cite{WuCHKLX14pvldb} and temporal reachability as follows.



\begin{definition}[Minimum Temporal Paths~\cite{WuCHKLX14pvldb}]    \label{de:minPath}
Let $\mathbf{P}(u,v,[t_{\alpha}, t_{\omega}])=\{P: P$ is a temporal path from $u$ to $v$ such that $start(P) \ge t_{\alpha}$, $end(P) \le t_{\omega}\}$.

A temporal path $P \in \mathbf{P}(u,v,[t_{\alpha}, t_{\omega}])$ is an \textbf{earliest-arrival path} if $end(P) = \min\{end(P'): P' \in \mathbf{P}(u,v,[t_{\alpha}, t_{\omega}])\}$. The \textbf{earliest-arrival time} to reach $v$ from $u$ within $[t_{\alpha}, t_{\omega}]$ is given by $end(P)$.

A temporal path $P \in \mathbf{P}(u,v,[t_{\alpha}, t_{\omega}])$ is a \textbf{fastest path} if $dura(P)$ $=$ $\min\{dura(P'): P' \in \mathbf{P}(u,v,[t_{\alpha}, t_{\omega}])\}$. The \textbf{minimum duration} taken to go from $u$ to $v$ within $[t_{\alpha}, t_{\omega}]$ is given by $dura(P)$.
\end{definition}


\begin{definition}[Temporal Reachability]    \label{de:reachability}
Given two vertices $u$ and $v$, and a time interval $[t_{\alpha}, t_{\omega}]$, $u$ can reach $v$ (or $v$ is reachable from $u$) within $[t_{\alpha}, t_{\omega}]$ if $\mathbf{P}(u,v,[t_{\alpha}, t_{\omega}])\neq \emptyset$, i.e., there exists a temporal path $P$ from $u$ to $v$ such that $start(P) \ge t_{\alpha}$ and $end(P) \le t_{\omega}$. 
\end{definition}

%
%


\begin{example} \label{eg:definition}
Figure~\ref{fig:transform}(a) shows a temporal graph $\mathcal{G}$. For simplicity, we assume that the traversal time for every edge is 1 time unit. In  $\mathcal{G}$, $a$ can reach $d$ within time interval $[2,5]$ since there is a temporal path $P=\langle(a, b, 2, 1),(b, d, 4, 1)\rangle$, while $a$ cannot reach $d$ within $[1,3]$ since there is no temporal path from $a$ to $d$ within $[1,3]$. Given source vertex $a$, target vertex $d$, and a time interval $[1, 10]$, $P_1=\langle(a, b, 2, 1),(b, d, 4, 1)\rangle$ is an earliest-arrival path with arrival time $4+1=5$, $P_2=\langle(a, c, 4, 1),(c, d, 5, 1)\rangle$ is a fastest path with duration $(5+1)-4=2$. Note that $P_1$ is not a fastest path, and  $P_2$ is not an earliest-arrival path within $[1, 10]$.
\end{example}

\noindent \textbf{Problem definition:} Given a temporal graph $\mathcal{G}$, we propose to construct an index, such that given a source vertex $u$ and a target vertex $v$, and a time interval $[t_{\alpha}, t_{\omega}]$, we can efficiently answer the following queries: (1) whether $u$ can reach $v$ within $[t_{\alpha}, t_{\omega}]$, (2) the earliest-arrival time going from $u$ to $v$ within $[t_{\alpha}, t_{\omega}]$, and (3) the minimum duration taken to go from $u$ to $v$ within $[t_{\alpha}, t_{\omega}]$.

Our method can also be applied to compute another type of minimum temporal path called \emph{\textbf{latest-departure path}}~\cite{WuCHKLX14pvldb}. However, the concept of latest-departure path is symmetric to that of earliest-arrival path.  We hence omit the details.



\section{Graph Transformation}  \label{sec:transform}


In addition to indexing and querying efficiency, another key consideration in designing an indexing method for querying temporal graphs is that the index must support efficient dynamic update. This is important since temporal edges are created and added frequently over time in most real applications (e.g., there can be many phone calls and text messages created just over a short period of time). Existing work such as TTL~\cite{WangLYXZ15sigmod}, however, does not support dynamic update. We carefully examined a graph transformation method proposed in~\cite{WuCHKLX14pvldb}, and found that the transformation can be applied to design an efficient index that also allows efficient dynamic update.

We first present how to transform a temporal graph $\mathcal{G}=(\mathcal{V},\mathcal{E})$ into a new graph $G=(V,E)$. The construction of $G$ consists of two phases:


\begin{enumerate}
  \item Construction of vertex set: for each vertex $v \in \mathcal{V}$, create vertices in $V$ as follows.


  		\begin{enumerate}
  			\item Let $T_{in}(u,v)=\{t+\lambda: (u,v,t,\lambda) \in \Pi(u,v)\}$ for each $u$$\in$$\Gamma_{in}(v,\mathcal{G})$, and $\mathbf{T}_{in}(v)$$=$$\bigcup_{u \in \Gamma_{in}(v,\mathcal{G})} T_{in}(u,v)$, i.e., $\mathbf{T}_{in}(v)$ is the set of distinct time instances at which edges from in-neighbors of $v$ arrive at $v$.

            Create $|\mathbf{T}_{in}(v)|$ copies of $v$, each labeled with $\langle v,t \rangle$ where $t$ is a distinct arrival time in $\mathbf{T}_{in}(v)$. Denote this set of vertices as $V_{in}(v)$, i.e., $V_{in}(v)=\{\langle v,t \rangle: t \in \mathbf{T}_{in}(v)\}$. Sort vertices in $V_{in}(v)$ in descending order of their time, i.e., for any $\langle v,t_1 \rangle, \langle v,t_2 \rangle \in V_{in}(v)$, $\langle v,t_1 \rangle$ is ordered before $\langle v,t_2 \rangle$ in $V_{in}(v)$ iff $t_1 > t_2$.

  			\item Let $T_{out}(v,u)=\{t: (v,u,t,\lambda) \in \Pi(v,u)\}$ for each $u \in \Gamma_{out}(v,\mathcal{G})$, and $\mathbf{T}_{out}(v)$$=$$\bigcup_{u \in \Gamma_{out}(v,\mathcal{G})} T_{out}(v,u)$.

            Create $|\mathbf{T}_{out}(v)|$ copies of $v$, each labeled with $\langle v,t \rangle$ where $t$ is a distinct starting time in $\mathbf{T}_{out}(v)$. Denote this set of vertices as $V_{out}(v)$, i.e., $V_{out}(v)=\{\langle v,t \rangle: t \in \mathbf{T}_{out}(v)\}$. Sort vertices in $V_{out}(v)$ in descending order of their time.

  		\end{enumerate}


  \item Construction of edge set: create edges in $E$ as follows.


  		\begin{enumerate}
  			\item Let $V_{in}(v)=\{\langle v,t_1 \rangle, \langle v,t_2 \rangle, \ldots, \langle v,t_h \rangle\}$, where $t_i > t_{i+1}$ for $1 \le i < h$. Create a directed edge from each $\langle v,t_{i+1} \rangle$ to $\langle v,t_i \rangle$, for $1 \le i < h$. No edge is created if $h \le 1$. Create edges for $V_{out}(v)$ in the same way.

  			\item For each vertex $\langle v,t_{in} \rangle$ according to its order in $V_{in}(v)$, create a directed edge from $\langle v,t_{in} \rangle$ to vertex $\langle v,t_{out} \rangle$ $\in V_{out}(v)$, where $t_{out}= \min\{ t : \langle v,t \rangle \in V_{out}(v), t \geq t_{in} \}$ and no edge from another vertex $\langle v,t'_{in} \rangle \in V_{in}(v)$ to $\langle v,t_{out} \rangle$ has been created.

  			\item For each temporal edge $e = (u, v, t, \lambda) \in \mathcal{E}$, create a directed edge from the vertex $\langle u, t \rangle \in V_{out}(u)$ to the vertex $\langle v, t+\lambda \rangle \in V_{in}(v)$.
		\end{enumerate}  		

\end{enumerate}


%


The following example illustrates graph transformation.


\begin{example} \label{eg:transform}
Figure~\ref{fig:transform}(b) shows the transformed graph $G$ of the temporal graph $\mathcal{G}$ in Figure~\ref{fig:transform}(a), where $\lambda=1$ for all edges. Although there is a cycle $\langle(a, c, 4, 1),$ $(c, a, 6, 1)\rangle$ in $\mathcal{G}$,  $G$ is a DAG.
\end{example}

\section{Top-k Chain Labeling}  \label{sec:index}

In this section, we present the \emph{top-$k$ chain labeling scheme}, which we name as \textbf{TopChain}. We focus on our discussion on reachability queries, and we discuss how TopChain is applied to answer time-based queries in Section~\ref{ssec:timequery}.

Many labeling schemes have been proposed to answer reachability queries~\cite{ChengHWF13sigmod,JinW13pvldb,SeufertABW13icde,TrisslL07sigmod,WeiYLJ14pvldb,YanoAIY13cikm,YildirimCZ12vldb,ZhuLWX14sigmod}. They first transform the input directed graph into a \emph{directed acyclic graph} (\emph{DAG}) by collapsing each \emph{strongly connected component} (\emph{SCC}) into a super vertex, and then construct labels on the much smaller DAG. Theoretically, the existing labeling schemes can also be applied to the transformed non-temporal graph of a temporal graph. However, practically there may be scalability issues since the transformed graph is a DAG itself and is much larger than the DAGs indexed by existing methods. In the experiments, we used real temporal graphs with more than 200 million edges, for which the transformed graph is more than an order of magnitude larger than the size of existing DAGs. Before we introduce our labeling scheme, we first prove that a transformed graph is a DAG.


Given a transformed graph $G=(V,E)$, for each vertex $\langle w,t \rangle \in V$, let $v=\langle w,t \rangle$. For simplicity, we often simply use $v$ instead of $\langle w,t \rangle$ in our discussion. We use $\langle w,t \rangle$ only when we need to refer to $v$'s \emph{time stamp} $t$. 


\begin{lemma}   \label{le:dag}
Let $G=(V,E)$ be the transformed graph of a temporal graph $\mathcal{G}$, if the traversal time of each edge in $\mathcal{G}$ is non-zero, then $G$ is a DAG.
\end{lemma}

\begin{proof}
We prove by contradiction. If $G$ is not a DAG, then there exists a cycle $C=\langle v=u_1, \ldots, u_h=v \rangle$ in $G$, where $h>1$ and $(u_{i}, u_{i+1})$ is an edge in $E$ for $1 \leq i < h$, i.e. $v$ can reach itself. Let $u_i=\langle w_i,t_i \rangle$ for $1 \leq i \le h$. According to the construction of $G$, for every edge $(u_{i}, u_{i+1})$, we have $t_i \le t_{i+1}$, or $t_i < t_{i+1}$ if $w_i \ne w_{i+1}$, since the traversal time of each edge is non-zero. If $w_1 = w_2 = \dots = w_h$, i.e., $w_i$ (for $1 \le i \le h$) is the same vertex in $\mathcal{G}$, then $C$ cannot be a cycle according to the construction of $G$. Thus, there exists some $i$ such that $w_i \ne w_{i+1}$, which implies that $t_1 \le \dots \le t_i < t_{i+1} \le \dots t_h$. But $v=u_1=u_h$ implies $t_1=t_h$, which is a contradiction.
\end{proof}


The size of $G$ is $O(|\mathcal{V}|+|\mathcal{E}|)$, and is even considerably larger than $\mathcal{G}$ (see Table~\ref{tab:realdata}). Thus, the DAGs we handle in this paper are significantly larger than those used to test existing indexes. In the following, we propose a more efficient method that can handle much larger DAGs.

\subsection{Label Construction} \label{ssec:labeling}

Given a DAG (not necessarily a transformed graph), $G=(V,E)$,  we use $u \rightarrow v$ to denote that $u$ can reach $v$ in $G$, and $u \nrightarrow v$ if $u$ cannot reach $v$. We assume $v \rightarrow v$ for any $v \in V$. We define the set of \textbf{out-neighbors} of a vertex $u$ in $G$ as $\Gamma_{out}(u,G)=\{v: (u, v) \in E\}$, and the \textbf{out-degree} of $u$ in $G$ as $d_{out}(u,G)=|\Gamma_{out}(u,G)|$. Similarly, we define $\Gamma_{in}(u,G)=\{v: (v, u) \in E\}$ and $d_{in}(u,G)=|\Gamma_{in}(u,G)|$.


TopChain constructs two sets of labels for each $v \in V$, denoted by $L_{in}(v)$ and $L_{out}(v)$, called \textbf{in-labels} and \textbf{out-labels} of $v$. TopChain takes an input parameter $k$, which is used to control the label set size, label construction time and query processing time.

The main idea is to compute the top $k$ labels based on a ranking defined on a chain cover of the input DAG $G=(V,E)$, for  $L_{in}(v)$ and $L_{out}(v)$ of each $v \in V$. Here, a \emph{chain} $C$ is an ordered sequence of $h$ vertices $C = \langle v_1, v_2, \ldots, v_h \rangle$, such that $v_i \rightarrow v_{i+1}$ for $1 \le i < h$. A \emph{chain cover} of $G$ is a disjoint partition $\{C_1, \ldots, C_l\}$ of $V$, where $C_i$ is a chain for $1 \le i \le l$.

\subsubsection{Definition of Labels}    \label{sssec:def_label}

Given a chain cover $\mathbb{C}=\{C_1, \ldots, C_l\}$ of $G$, and a rank $rank(C)$ for each chain $C \in \mathbb{C}$, the labels of each $v \in V$ are defined as follows.

Define a \textbf{chain code} for $v$ as $code(v)=(v.x, v.y)$, where $v.x=rank(C)$ and $v$ in $C$, and $v.y$ is the position of $v$ in $C$. Let $RC(v) \subseteq \mathbb{C}$ be the set of chains that $v$ can reach; formally, for each $C \in RC(v)$, there exists a vertex $u$ in $C$ such that $v \rightarrow u$. For each $C \in RC(v)$, let $first_v(C)$ be the first vertex in $C$ that $v$ can reach, i.e., if $C = \langle v_1, v_2, \ldots, v_h \rangle$ and $first_v(C)=v_i$, then $v \rightarrow v_j$ for $i \le j \le h$ and $v \nrightarrow v_{j'}$ for $1 \le j' < i$. Note that $first_v(C)=v$, where $v$ in $C$, since we assume $v \rightarrow v$.

Define $RF(v)$ as the set of first reachable vertices in the set of reachable chains from $v$, i.e., $RF(v)=\{first_v(C): C \in RC(v)\}$. Define $RFcode(v) = \{code(u): u \in RF(v)\}$. Sort the chain codes in $RFcode(v)$ in ascending order of $u.x$ for each $(u.x, u.y) \in RFcode(v)$, and let $top_{k}(RFcode(v))$ be the first $k$ chain codes in $RFcode(v)$ and $top_{k}(RFcode(v)) = RFcode(v)$ if $|RFcode(v)|$$\le$$k$. Then, $L_{out}(v)$$=$$top_{k}(RFcode(v))$.

Similarly, let $RC^{-1}(v) \subseteq \mathbb{C}$ be the set of chains that can reach $v$, $last_v(C)$ be the last vertex in $C$ that can reach $v$. Define $RL(v)=\{last_v(C): C \in RC^{-1}(v)\}$. Define $RLcode(v) = \{code(u): u \in RL(v)\}$. Sort the chain codes in $RLcode(v)$ in ascending order of $u.x$ for each $(u.x, u.y) \in RLcode(v)$, and let $top_{k}(RLcode(v))$ be the first $k$ chain codes in $RLcode(v)$ and $top_{k}(RLcode(v)) = RLcode(v)$ if $|RLcode(v)| \le k$. Then, $L_{in}(v)=top_{k}(RLcode(v))$.

\begin{example} \label{eg:def_label}
Figure~\ref{fig:dag} shows a DAG $G$, which is in fact $G$ in Figure~\ref{fig:transform}(b) by relabeling the vertices, and a chain cover of $G$, $\mathbb{C}=\{C_1, C_2, C_3, C_4\}$, where $C_1=\langle v_1,v_2,v_3,v_4 \rangle$, $C_2=\langle v_5,v_6,v_7 \rangle$, $C_3=\langle v_8,v_9,v_{10} \rangle$, $C_4=\langle v_{11},v_{12} \rangle$. The chain code of $v_3$ is $code(v_3)=(1,3)$ since $v_3$ is at the 3-rd position in $C_1$. And $RC(v_3)=\{C_1, C_3, C_4\}$, since $v_3$ can reach $C_1$, $C_3$ and $C_4$. $RF(v_3)=\{v_3,v_8,v_{12}\}$, $RFcode(v_3)=\{(1,3),(3,1),(4,2)\}$, since $code(v_3)=(1,3)$, $code(v_8)=(3,1)$, $code(v_{12})=(4,2)$. Let $k=2$, then $L_{out}(v_3)=top_{2}(RFcode(v_3))$$=$$\{(1,3),(3,1)\}$. Similarly, $RC^{-1}(v_3)$$=$$\{C_1\}$, $RLcode(v_3)$$=$$\{(1,3)\}$, $L_{in}(v_3)$$=$$top_{2}(RLcode(v_3))$$=$$\{(1,3)\}$.
\end{example}

\begin{figure}[!tbp]
\begin{center}
\includegraphics[width = 2.5in]{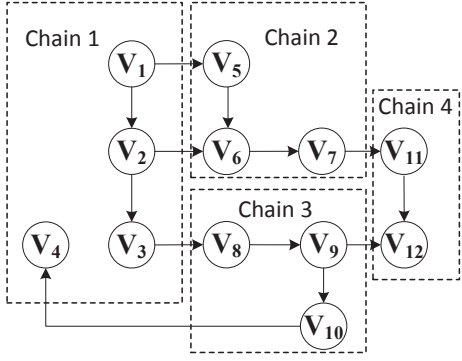}
\caption{Chain cover of $G$ in Figure~\ref{fig:transform}(b)} \label{fig:dag}
\end{center}
\end{figure}


\subsubsection{Algorithm for Labeling}    \label{sssec:alg_label}


We now present our method to compute the above-defined labels, as outlined in Algorithm~\ref{alg:labeling}. We discuss how to compute a chain cover of $G$ and chain ranking in Section~\ref{ssec:chain}. Given a chain cover $\mathbb{C}$ and a rank $rank(C)$ for each chain $C \in \mathbb{C}$, Lines~\ref{ln:selfassign0}-\ref{ln:selfassign} first assign the chain code of each vertex $v$ to both $L_{out}(v)$ and $L_{in}(v)$.



Assume that each vertex in $V$ is assigned a topological order, which can be obtained by topological sort on $G$. Then, Lines~\ref{ln:outlabel0}-\ref{ln:outlabel} compute $L_{out}(v)$ for each $v \in V$ in reverse topological order. For each $v$, the algorithm computes the top $k$ labels with the smallest chain rank from the set of all out-labels of $v$'s out-neighbors and the out-label of $v$. This can be done by scanning $L_{out}(u)$ for each $u \in \Gamma_{out}(v, G)$ as the merge phase in merge-sort until the top $k$ labels are obtained, since the labels in each $L_{out}(u)$ are ordered according to chain rank. In addition, Line~\ref{ln:topk} ensures that at most one vertex from each chain can have its chain code included as one of the top $k$ labels. Then, the top $k$ labels are assigned to $L_{out}(v)$. Similarly, $L_{in}(v)$ for each $v \in V$ is computed in Lines~\ref{ln:inlabel0}-\ref{ln:inlabel}.

\begin{algorithm}[!t]

\SetKwInOut{input}{Input}\SetKwInOut{output}{Output}

{\small

\input{A DAG $G=(V,E)$, a chain cover $\mathbb{C}=\{C_1, \ldots, C_l\}$ of $G$, where $C_i$ is ranked before $C_j$ for $1 \le i < j \le l$, and an integer $k$}
\output{$L_{out}(v)$ and $L_{in}(v)$ for every vertex $v \in V$}

    \ForEach {vertex, $v \in V$,}{  \label{ln:selfassign0}
        Let $C$ be the chain that contains $v$\;
        Assign a chain code $(v.x, v.y)$ to $v$, where $v.x=rank(C)$ and $v.y$ is the position of $v$ in $C$\;
        $L_{in}(v) \gets \{ (v.x, v.y)\}$, $L_{out}(v) \gets \{ (v.x, v.y)\}$\;   \label{ln:selfassign}
    }

    \ForEach {vertex, $v \in V$, in reverse topological order}{ \label{ln:outlabel0}

        Let $L=L_{out}(v) \cup \bigcup_{u \in \Gamma_{out}(v, G)} L_{out}(u)$\;
        Let $L_k$ be the top $k$ labels with the smallest chain rank from $L$ such that: if there exist two labels $(u.x, u.y)$ and $(w.x, w.y)$ in $L$ such that $u.x = w.x$ and $u.y < w.y$, i.e., $u$ and $w$ are in the same chain, then $\underline{(w.x, w.y)} \notin L_k$\;    \label{ln:topk}
        $L_{out}(v) \gets L_k$\;     \label{ln:outlabel}

    }

    \ForEach {vertex, $v \in V$, in topological order}{  \label{ln:inlabel0}

        Let $L=L_{in}(v) \cup \bigcup_{u \in \Gamma_{in}(v, G)} L_{in}(u)$\;
        Let $L_k$ be the top $k$ labels with the smallest chain rank from $L$ such that: if there exist two labels $(u.x, u.y)$ and $(w.x, w.y)$ in $L$ such that $u.x = w.x$ and $u.y < w.y$, i.e., $u$ and $w$ are in the same chain, then $\underline{(u.x, u.y)} \notin L_k$\;
        $L_{in}(v) \gets L_k$\; \label{ln:inlabel}
    }
}
\caption{TopChain: Label Construction}
\label{alg:labeling}	
\end{algorithm}



\begin{lemma} \label{le:labeling}
Algorithm~\ref{alg:labeling} correctly computes $L_{out}(v)$ and $L_{in}(v)$ for every vertex $v \in V$.
\end{lemma}

\begin{proof}
We first prove that Algorithm~\ref{alg:labeling} correctly computes $L_{out}(v)$ for every $v \in V$. Let $RF_k(v) = \{u: (u.x, u.y) \in top_{k}(RFcode(v))\}$. We prove that for every vertex $u \in RF_k(v)$, $(u.x, u.y) \in L_{out}(v)$ when Algorithm~\ref{alg:labeling} terminates.

Since $u \in RF_k(v)$ implies $v \rightarrow u$, there exists a path $\langle v=u_1, u_2, \ldots, u_p = u \rangle$ in $G$. First, we want to prove that $u \in RF_k(u_i)$ for $1 \le i \le p$. We prove by contradiction as follows. Assume to the contrary $u \notin RF_k(u_i)$ for some $i$ where $1 \le i \le p$. Then, since $v \rightarrow u_i$ implies $RC(u_i) \subseteq RC(v)$, by the definition of $RF_k(.)$, and $u_i \rightarrow u$, we have $u \notin RF_k(v)$, which is a contradiction and thus $u \in RF_k(u_i)$ for $1 \le i \le p$.

Then, we show $(u.x, u.y) \in L_{out}(u_i)$ for $1 \le i \le p$. Since $u \in RF_k(u_p)=RF_k(u)$, we have $(u.x, u.y) \in L_{out}(u)$, which is computed in Line~\ref{ln:selfassign}. Similarly, when we consider vertex $u_{p-1}$, Algorithm~\ref{alg:labeling} computes $L_{out}(u_{p-1})$ by merging it with $L_{out}(u=u_p)$, and thus $(u.x, u.y)$ will be added to $L_{out}(u_{p-1})$ since $u \in RF_k(u_{p-1})$. Again, since $u \in RF_k(u_{p-1})$, $(u.x, u.y)$ will be kept in the final $L_{out}(u_{p-1})$. Repeating the analysis for $u_i$, where $1 \le i < p-1$, we can conclude that $(u.x, u.y) \in L_{out}(u_i)$, for $1 \le i \le p$. Thus, $(u.x, u.y) \in L_{out}(v)$ for every $u \in RF_k(v)$.

Since Algorithm~\ref{alg:labeling} keeps at most $k$ labels for $L_{out}(v)$, Algorithm~\ref{alg:labeling} correctly computes $L_{out}(v) = top_{k}(RFcode(v))$ if $|RFcode(v)|$ $\ge k$. If $|RFcode(v)| < k$, we need to show that for any $w \notin RF_k(v)$, $(w.x, w.y) \notin L_{out}(v)$. Suppose to the contrary that $(w.x, w.y) \in L_{out}(v)$ as computed by Algorithm~\ref{alg:labeling}, then $w$ must be a descendant of $v$ and hence $v \rightarrow w$. Since $w \notin RF_k(v)$ and $v \rightarrow w$, there must be another vertex $u \in RF_k(v)$ such that $u$ and $w$ are in the same chain $C$ and $u$ is the first vertex in $C$ that is reachable from $v$, i.e., $u.x = w.x$ and $u.y < w.y$. However, according to Line~\ref{ln:topk}, $(w.x, w.y)$ will not be added to $L_{out}(v)$ in this case. Thus, Algorithm~\ref{alg:labeling} also correctly computes $L_{out}(v) = top_{k}(RFcode(v))$ when $|RFcode(v)| < k$.

Similarly, we can prove that Algorithm~\ref{alg:labeling} correctly computes $L_{in}(v)$ for every $v \in V$.
\end{proof}


\begin{theorem} \label{th:labeling}
Given a DAG $G=(V,E)$ and a chain cover of $G$, Algorithm~\ref{alg:labeling} correctly computes $L_{out}(v)$ and $L_{in}(v)$ for all $v \in V$ in $O(k(|V|+|E|))$ time, and the total label size is given by $\sum_{v \in V}(|L_{out}(v)| + |L_{in}(v)|) = O(k|V|)$.
\end{theorem}



\begin{proof}
The time complexity of Algorithm~\ref{alg:labeling} is dominated by the two for-loops in Lines~\ref{ln:outlabel0}-\ref{ln:outlabel} and Lines~\ref{ln:inlabel0}-\ref{ln:inlabel}. Topological sort can be computed in $O(|V|+|E|)$ time. The cost of Line~\ref{ln:topk} is $O(k \times d_{out}(v,G))$, since each $L_{out}(u)$ is bounded by $k$. Thus, the total cost of Lines~\ref{ln:outlabel0}-\ref{ln:outlabel} is given by $\sum_{v \in V} O(k \times d_{out}(v,G)) = O(k |E|)$. Similarly, the cost of Lines~\ref{ln:inlabel0}-\ref{ln:inlabel} is also $O(k |E|)$. Summing up, the total complexity is given by $O(k(|V|+|E|))$.

Since for every $v \in V$, $|L_{in}(v)| \le k$ and $|L_{out}(v)| \le k$, the total label size is bounded by $O(k|V|)$.
\end{proof}

We show how Algorithm~\ref{alg:labeling} computes the labels as follows.

\begin{example} \label{eg:alg_label}
Given the chain cover $\mathbb{C}=\{C_1, C_2, C_3, C_4\}$ of $G$ in Figure~\ref{fig:dag}, and $k=2$, Algorithm~\ref{alg:labeling} first initializes $L_{in}(v)$ and $L_{out}(v)$ to contain $v$ itself. One topological order is as follows: $\langle v_1,v_2,v_3,v_5,v_6,v_7,v_8,v_9,v_{10},v_4,v_{11},v_{12} \rangle$. Then, in reverse topological order, we compute $L_{out}(v)$ for every $v$. We show how $L_{out}(v_{3})$ is computed. First, we compute $L_{out}(v_{12})=code(v_{12})$ $=\{(4,2)\}$ and $L_{out}(v_4)=code(v_4)=\{(1,4)\}$. Next, we compute $L_{out}(v_{10})=code(v_{10}) \cup L_{out}(v_4) =\{(1,4),(3,3)\}$, $L_{out}(v_{9})=top_2(code(v_9) \cup L_{out}(v_{10})$ $\cup L_{out}(v_{12}))=top_2\{(3,2),$ $(1,4),(3,3),(4,2)\}=\{(1,4),(3,2)\}$, $L_{out}(v_{8})=top_2(code(v_8) \cup L_{out}(v_{9})) =\{(1,4),(3,1)\}$, $L_{out}(v_{3})$ $=top_2(code(v_3) \cup L_{out}(v_{8}))$ $=\{(1,3),(3,1)\}$. Similarly, we compute $L_{in}(v_{3})=\{(1,3)\}$.
\end{example}

\subsection{Chain Cover and Chain Ranking} \label{ssec:chain}


One input to Algorithm~\ref{alg:labeling} is a chain cover of $G$. An optimal chain cover, which is one that consists of the minimum number of chains, can be computed by a min-flow based method in $O(|V|^{3})$ time~ \cite{Jagadish90tods}. And the time is reduced to $O(|V|^{2}+|V|l\sqrt{l})$ based on bipartite matching~\cite{ChenC08icde}, where $l$ is the number of chains. Both of these two methods are too expensive for processing large graphs.

Since our labeling scheme presented in Section~\ref{ssec:labeling} is not limited to a transformed graph but any general DAG, we apply a greedy algorithm~\cite{Simon88tcs} to compute a chain cover if the application is to answer reachability queries in a non-temporal graph. The greedy algorithm grows a chain by recursively adding the smallest-ranked out-neighbor of the last vertex in the chain, where the ranking is defined based on a topological ordering of the vertices. The algorithm uses $O(|V|\log d_{max} + |E|)$ time, where $d_{max}$ is the maximum degree of a vertex in the DAG. 

For processing a temporal graph $\mathcal{G}$, we adopt a simple and efficient method based on the property of $\mathcal{G}$ as follows. In the transformed graph $G=(V,E)$ of $\mathcal{G}=(\mathcal{V},\mathcal{E})$, each set of vertices $V_{in}(v)$ or $V_{out}(v)$ naturally appears as a chain. Thus, we can obtain a natural chain cover of $G$, i.e., $\mathbb{C}=\{V_{in}(v): v \in \mathcal{V}\} \cup \{V_{out}(v): v \in \mathcal{V}\}$. In fact, we can reduce the number of chains in $\mathbb{C}$ by half as follows. We merge $V_{in}(v)$ and $V_{out}(v)$ into one single chain in ascending order of the time stamp of the vertices. If there exist $\langle v,t_{in} \rangle \in V_{in}(v)$ and $\langle v,t_{out} \rangle \in V_{out}(v)$ such that $t_{in}=t_{out}$, we order $\langle v,t_{in} \rangle$ before $\langle v,t_{out} \rangle$. 


Merging $V_{in}(v)$ and $V_{out}(v)$ into a single chain $C$ can approximately reduce query response time by half, since logically $V_{in}(v)$ and $V_{out}(v)$ belong to a single vertex $v$ in the original temporal graph. For example, the chain cover in Figure~\ref{fig:dag} is constructed in this way. 


However, theoretically there is one small problem, which can be easily fixed, though we need to show a rigorous proof to show the correctness of our indexing method. We first present the problem as follows. Let $C=\langle u_1, u_2, \ldots, u_h \rangle$. The definition of chain requires that $u_i \rightarrow u_j$ for $1\le i < j \le h$. However, in a temporal graph $\mathcal{G}$, it is possible that a vertex cannot reach itself, i.e., $u_i$ may not reach $u_j$ in the transformed graph $G$, for some $u_i \in V_{out}(v)$ and $u_j \in V_{in}(v)$. Essentially, merging $V_{in}(v)$ and $V_{out}(v)$ into a single chain $C$ creates a new graph $G_{new}$, by adding an edge from $\langle v,t_{out} \rangle \in V_{out}(v)$ to $\langle v,t_{in} \rangle \in V_{in}(v)$ to $\mathcal{G}$ for each $v \in \mathcal{V}$ if $t_{out} < t_{in}$.

However, $G_{new}$ only exists conceptually used to define the chain cover, and we never really use $G_{new}$ in our algorithm for label construction. Note that we do not compute the chain cover from $G_{new}$, but simply form a chain $C$ from each $V_{in}(v)$ and $V_{out}(v)$ in $G$, although the reachability of the vertices in $C$ is defined based on $G_{new}$ instead of $G$.

In the following theorem, we show the correctness of Algorithm~\ref{alg:labeling} when the input is $G$ and the chain cover is based on $G_{new}$, even though there exists false reachability information in $G_{new}$. In Section~\ref{ssec:timequery}, we will also show that the labels give correct answers to time-based queries.


 %




%



\begin{theorem} \label{th:mergechain}




Let $\mathbb{C}$ be the chain cover defined based on $G_{new}$, where each $V_{in}(v)$ and $V_{out}(v)$ in $G_{new}$ are merged into a single chain $C \in \mathbb{C}$. Given $G$ and $\mathbb{C}$ as input, Algorithm~\ref{alg:labeling} constructs the same labels as the labels defined based on $G_{new}$ and $\mathbb{C}$ (i.e., constructed by Algorithm~\ref{alg:labeling} with $G_{new}$ and $\mathbb{C}$ as input).
\end{theorem}


\begin{proof}
Let $C = \langle u_1, u_2, \ldots, u_h \rangle$ be the chain obtained by merging $V_{in}(v)$ and $V_{out}(v)$, for some $v$ in the original temporal graph $\mathcal{G}$, such that there exist $u_i$ and $u_j$, where $1 \le i < j \le h$, $u_i$ cannot reach $u_j$ in $G$. We show that the false reachability information ``$u_i \rightarrow u_j$'' presented in $C$ does not change $RF(r)$ and $RL(r)$ for any vertex $r \in V$.

Let $r = \langle a, t_a \rangle$, and consider another vertex $s=\langle b, t_b \rangle \in V$, where $a \ne b$. We first prove that $r$ can reach $s$ (or $r$ is reachable from $s$) in $G$ if and only if $r$ can reach $s$ (or $r$ is reachable from $s$) in $G_{new}$, i.e., the false information ``$u_i \rightarrow u_j$'' presented in $C$ does not affect the reachability between $r$ and $s$ in $G$.

We first show that if $r$ and $s$ are not in $C$, then the false information ``$u_i \rightarrow u_j$'' presented in $C$ does not affect the reachability between $r$ and $s$ in $G$. According to the construction of $G$, only vertices in $V_{in}(v)$ can have in-neighbors that are not vertices in $C$, while only vertices in $V_{out}(v)$ can have out-neighbors that are not vertices in $C$. Thus, we only need to show that for any $u_{in} = \langle v, t_{in} \rangle \in V_{in}(v)$ and any $u_{out} = \langle v, t_{out} \rangle \in V_{out}(v)$, $u_{in}$ is ordered before $u_{out}$ in $C$ if and only if $u_{in}$ can reach $u_{out}$ in $G$. First, if $u_{in}$ can reach $u_{out}$ in $G$, then $t_{in} \le t_{out}$ according to the construction of $G$, which means that $u_{in}$ is ordered before $u_{out}$ in $C$. Next, if $u_{in}$ is ordered before $u_{out}$ in $C$, then $t_{in} \le t_{out}$ since $V_{in}(v)$ and $V_{out}(v)$ are merged in ascending order of the time stamp of the vertices. According to the construction of the transformed graph $G$, for any vertex $\langle v, t_1 \rangle \in V_{in}(v)$ and for any $\langle v, t_2 \rangle \in (V_{in}(v) \cup V_{out}(v))$, if $t_1 \le t_2$, there is a path from $\langle v, t_1 \rangle$ to $\langle v, t_2 \rangle$ in $G$. Thus, $u_{in}$ can reach $u_{out}$ in $G$.

Now we consider the case that only $r$ is in $C$. If $r \in V_{in}(v)$, then the same result follows from the above analysis. If $r \in V_{out}(v)$ and $r \rightarrow s$, then clearly $r \rightarrow s$ regardless of ``$u_i \rightarrow u_j$''. If $r \in V_{out}(v)$ and $r \nrightarrow s$, then ``$u_i \rightarrow u_j$'' also does not give $r \rightarrow s$, which we prove by contradiction as follows. Suppose now $r \rightarrow s$ because of ``$u_i \rightarrow u_j$''. The path from $u_j$ to $s$ must pass through a vertex $u_{j'} \in V_{out}(v)$, where $j \le j' \le h$. However, according to the construction of $G$, if $u_i$ cannot reach $u_j$ in $G$ and $i < j$, then $u_i \in V_{out}(v)$ and $u_j \in V_{in}(v)$. Thus, we have $u_i \rightarrow u_{j'}$ and hence $r \rightarrow s$ even if $u_i \nrightarrow u_j$, which is a contradiction.

The case that only $s$ is in $C$ can be proved similarly. And since $a \ne b$, $r$ and $s$ cannot be both in $C$. Thus, the false information ``$u_i \rightarrow u_j$'' presented in $C$ does not affect the reachability between $r$ and $s$ in $G$. Since $r \in RF(r)$ and $r \in RL(r)$ regardless of ``$u_i \rightarrow u_j$'', $RF(r)$ and $RL(r)$ remain unchanged.

Since the top $k$ labels are selected from $RF(.)$ and $RL(.)$, and the ranking of the chain is computed based on $\mathbb{C}$, we can conclude that Algorithm~\ref{alg:labeling} constructs the same labels as the labels defined based on $G_{new}$.
\end{proof}

This chain cover can be naturally computed at no extra cost during the process of graph transformation, and thus the whole process takes only linear time. In addition, for the chain code $(v.x, v.y)$ of each vertex $\langle v,t \rangle \in V$, instead of assigning $v.y$ as the position of $v$ in its chain, we can directly use the time stamp of $v$, i.e., $v.y=t$. This new assignment of $v.y$ is in fact significant when update maintenance of the labels is considered. There can be frequent edge insertions in a temporal graph over time and in this case the labels need to be updated as well. If $v.y$ is assigned as the position of $v$ in its chain, then updating the labels is more difficult since inserting a vertex in a chain affects the position of all following vertices in the chain, which can in turn affect the labels of a large number of vertices in the graph. On the other hand, if $v.y=t$, then we can simply insert the vertex in the chain and $u.y$ for any vertex $u$ following $v$ in the chain needs not be updated. Dynamic update of labels will be discussed in Section~\ref{ssec:update}.


Each chain in $\mathbb{C}$ is assigned a rank for labeling. There are many different strategies to rank the chains. We only discuss strategies with a low computation cost, that is, they are practical for large graphs. Two such strategies are discussed as follows.

\begin{itemize}
  \item Random ranking: We rank the chains randomly. We use this method as a baseline.
  \item Ranking by degree: Let $\Phi(C)$ denote the sum of out-edges and in-edges of all the vertices in a chain $C$. We rank the chains in descending order of their $\Phi$ value, where the top-ranked chain has a rank of 1. The rationale for this ranking is that the higher the value of $\Phi(C)$, the higher is the probability that $C$ can reach and are reachable from a larger set of vertices in $G$. Thus, assigning a top rank to $C$ enables more vertices to contain $rank(C)$ in their labels, thus allowing a more efficient query processing. We use this method in our TopChain method. We apply radix sort to sort the chains in order to assign ranks, and hence maintain the linear index construction time complexity.
\end{itemize}


\subsection{Dynamic Update of Labels} \label{ssec:update}


New edges and vertices may be added to a temporal graph $\mathcal{G}$ over time. Since adding an isolated vertex is trivial, we only discuss the addition of a new edge $e=(a, b, t, \lambda)$. We need to update $G$ by inserting $u=\langle a, t \rangle$ into $V_{out}(a)$ and $v=\langle b, t+\lambda \rangle$ into $V_{in}(b)$, and adding an edge from $u$ to $v$. Consequently, the labels should be updated as follows.

First, we need to insert $u$ into the chain $C$ that is formed from $V_{out}(a)$ and $V_{in}(a)$. If $C$ does not exist, we create $u$ as a new chain, assign it a rank $l$ that is larger than that of existing chains, and initialize $L_{out}(u)$$=$$L_{in}(u)$$=$$(l,t)$. If $C$ exists, we insert $u$ into the right position in $C$ according to $t$, and initialize $L_{out}(u)$$=$$L_{in}(u)$$=$ $(rank(C),t)$. Let $u_1$ be the vertex ordered before $u$ in $C$ and $u_2$ be the vertex ordered after $u$ in $C$. We compute $L_{in}(u)$ as the top $k$ labels from $L_{in}(u) \cup L_{in}(u_1)$, and $L_{out}(u)$ as the top $k$ labels from $L_{out}(u) \cup L_{out}(u_2)$. Similarly, we compute $L_{out}(v)$ and $L_{in}(v)$.

Second, after inserting a new edge $(u,v)$ into $G$, we update the labels as follows. We perform a reverse BFS starting from vertex $u$ in $G$ to update the out-labels of vertices that are visited, since only these vertices may change their out-labels. For any vertex $w$ visited, let $w'$ be the parent of $w$ in the reverse BFS, we update $L_{out}(w)$ as the top $k$ labels from $L_{out}(w) \cup L_{out}(w')$. If $L_{out}(w)$ remains unchanged, then we do not continue the search from $w$. Similarly, we conduct a BFS starting from vertex $v$ to update the in-labels of the visited vertices.

The algorithm completes label updating in $O(k(|V|+|E|))$ time, which is the optimal worst case time. In practice, the update is very efficient, as we demonstrate by experiments.



\section{Query Processing by TopChain}  \label{sec:query}

We now discuss how we use the labels constructed in Section~\ref{sec:index} to answer reachability queries and minimum temporal path queries.

\subsection{Reachability Queries}   \label{ssec:reachability}


We process a reachability query whether $u \rightarrow v$ as shown in Algorithm~\ref{alg:query}. We first define a few operators used in the algorithm.

\begin{algorithm}[!t]

\SetKwInOut{input}{Input}\SetKwInOut{output}{Output}

{\small
\input{A DAG $G=(V,E)$, $\mathbb{L}=\{(L_{out}(v),L_{in}(v)): v \in V\}p$, and a pair of query vertices $(u,v)$}
\output{The answer whether $u$ can reach $v$}
    \If {$u.x = v.x$}{  \label{ln:samechain0}
        \If{$u.y \le v.y$}{ \label{ln:samechain1}
            \Return true\; \label{ln:samechain2}
        }
        \Return false\; \label{ln:samechain}
    }
    \If {$L_{out}(u) \gg L_{out}(v)$ or $L_{in}(v) \gg L_{in}(u)$} {    \label{ln:gg0}
        \Return false\;     \label{ln:gg}
    }
    \If {$L_{out}(u) \oplus L_{in}(v) = 1$}{    \label{ln:oplus0}
        \Return true\;         \label{ln:oplus}
    }

    \ForEach {$w \in \Gamma_{out}(u,G)$}{ \label{ln:bfs0}
        \If{$w$ has been not visited}{
            \If{\textbf{ReachQ}$(G, \mathbb{L}, (w,v))$ returns true} {
                \Return true\;  \label{ln:bfs}
            }
        }
    }

    \Return false\;

}
\caption{\textbf{ReachQ}$(G, \mathbb{L}, (u,v))$: Reachability Querying}
\label{alg:query}	
\end{algorithm}

We first define the operator $\oplus$:
\begin{displaymath}
L_{out}(u) \oplus L_{in}(v) = \left\{ \begin{array}{ll}
1 & \begin{split} \textrm{if } \exists (r.x, r.y) \in L_{out}(u),\\ (s.x, s.y) \in L_{in}(v), \textrm{s.t.} \\ r.x = s.x \textrm{ and }  r.y \le s.y \\ \end{split} \\
0 & \textrm{otherwise}
\end{array} \right.
\end{displaymath}
Intuitively, $L_{out}(u) \oplus L_{in}(v)$ tests whether there exist two vertices $r$ and $s$, where $(r.x, r.y) \in L_{out}(u)$ and $(s.x, s.y) \in L_{in}(v)$, such that $r$ and $s$ are in the same chain, and either $r=s$ or $r$ is ordered before $s$ in the chain. The following lemma shows how the operator can be used in reachability query processing.


\begin{lemma} \label{le:oplus}
If $L_{out}(u) \oplus L_{in}(v)=1$, then $u \rightarrow v$.
\end{lemma}
\begin{proof}
If $L_{out}(u) \oplus L_{in}(v)=1$, then it implies $u \rightarrow r$, $r \rightarrow s$, $s\rightarrow v$, and thus $u \rightarrow v.$
\end{proof}


The following example illustrates how the operator $\oplus$ works.
\begin{example} \label{eg:oplus}
Consider $G$ in Figure ~\ref{fig:dag} and $k=2$, Algorithm~\ref{alg:labeling} computes $L_{out}(v_3)=\{(1,3),(3,1)\}$ and $L_{in}(v_{12})=\{(1,3),(3,2)\}$. Consider a reachability query that asks whether $v_3 \rightarrow v_{12}$. Since $L_{out}(v_{3}) \oplus L_{in}(v_{12})=1$ as $\exists (1, 3) \in L_{out}(v_3)$ and $(1, 3) \in L_{in}(v_{12})$, we conclude $v_3 \rightarrow v_{12}$.
\end{example}

Next, we define another operator $\gg$ as follows. We say $L_{out}(u) \gg L_{out}(v)$ if one of the following two cases is true:
\begin{itemize}
  \item Case (1): $\exists (r.x, r.y) \in L_{out}(v)$, $\nexists (w.x, w.y) \in L_{out}(u)$ such that $w.x = r.x$, and $\exists (s.x, s.y) \in L_{out}(u)$ such that $s.x > r.x$;
  \item Case (2): $\exists (r.x, r.y) \in L_{out}(v)$ and $(w.x, w.y) \in L_{out}(u)$ such that $w.x = r.x$ and $w.y > r.y$.
\end{itemize}


Intuitively, $L_{out}(u) \gg L_{out}(v)$ tests whether (1) there exists a vertex $r$ in a chain $C_1$ in $L_{out}(v)$, not exists any vertex in $L_{out}(u)$ in the same chain $C_1$, and exists a vertex $s$ in a chain $C_2$ in $L_{out}(u)$, such that the chain rank of $C_2$ is larger than that of $C_1$, which indicates that $v$ can reach at least one vertex in $C_1$, while $u$ cannot reach any vertex in $C_1$; or (2) there exists a vertex $r$ in a chain $C$ in $L_{out}(v)$, and a vertex $w$ in same chain $C$ in $L_{out}(u)$, such that $r$ is ordered before $w$ in $C$, which indicates that the first vertex in $C$ that $v$ can reach is $r$, the first vertex in $C$ that $u$ can reach is $w$, and $r \rightarrow w$.

Similarly, we say $L_{in}(v) \gg L_{in}(u)$ if one of the following two cases is true:


\begin{itemize}
  \item Case (1): $\exists (r.x, r.y) \in L_{in}(u)$, $\nexists (w.x, w.y) \in L_{in}(v)$ such that $w.x$$=$$r.x$, and $\exists (s.x, s.y) \in L_{in}(v)$ such that $s.x$$>$$r.x$;


  \item Case (2): $\exists (r.x, r.y) \in L_{in}(u)$ and $(w.x, w.y) \in L_{in}(v)$ such that $w.x = r.x$ and $w.y < r.y$.
\end{itemize}



The following lemma shows how the operator $\gg$ can be used in reachability query processing.

\begin{lemma} \label{le:gg}
If $L_{out}(u) \gg L_{out}(v)$ or $L_{in}(v) \gg L_{in}(u)$, then $u \nrightarrow v$.
\end{lemma}

\begin{proof}
First, we prove if $L_{out}(u) \gg L_{out}(v)$, then $u \nrightarrow v$. We prove by contradiction, by assuming that $L_{out}(u) \gg L_{out}(v)$ and $u \rightarrow v$. Consider the two cases of $L_{out}(u) \gg L_{out}(v)$.

If Case (1) is true, then: since $(s.x, s.y) \in L_{out}(u)$, we have $(s.x, s.y) \in top_{k}(RFcode(u))$. Since $\nexists (w.x, w.y) \in L_{out}(u)$ s.t. $w.x = r.x$, we have: $\nexists (w.x, w.y) \in top_{k}(RFcode(u))$ s.t. $w.x = r.x$. Then, $(r.x, r.y) \in L_{out}(v)$ implies $v \rightarrow r$, and together with the assumption $u \rightarrow v$, it implies $u \rightarrow r$. However, by the definition of $top_{k}(RFcode(u))$, $u \rightarrow r$ and $(s.x, s.y) \in top_{k}(RFcode(u))$, where $s.x > r.x$, implies that there must exist $(w.x, w.y) \in top_{k}(RFcode(u))$ s.t. $w.x = r.x < s.x$. This is a contradiction and hence $L_{out}(u) \gg L_{out}(v)$ must imply $u \nrightarrow v$.

If Case (2) is true, then: since $(w.x, w.y) \in L_{out}(u)$, we have $w \in RF(u)$. Then, $(r.x, r.y) \in L_{out}(v)$ implies $v \rightarrow r$, and together with the assumption $u \rightarrow v$, it implies $u \rightarrow r$. However, by the definition of $RF(u)$, $u \rightarrow r$ and $w \in RF(u)$, where $w.x = r.x$ and $w.y > r.y$, implies that $r$ should be in $RF(u)$ instead of $w$. Thus, we have a contradiction and hence $L_{out}(u) \gg L_{out}(v)$ must imply $u \nrightarrow v$.

Similarly, we can show that if $L_{in}(v) \gg L_{in}(u)$, then $u \nrightarrow v$.
\end{proof}


We illustrate how the operator $\gg$ works as follows.


\begin{example} \label{eg:gg}
Consider $G$ in Figure ~\ref{fig:dag} and $k=2$, Algorithm~\ref{alg:labeling} computes $L_{out}(v_{2})=\{(1,2),(2,2)\}$, $L_{out}(v_3)=\{(1,3),(3,1)\}$ and $L_{out}(v_{5})=\{(2,1),(4,1)\}$. Since $\exists (3, 1) \in L_{out}(v_3)$,  $\nexists (w.x, w.y)$ $\in L_{out}(v_5)$ such that $w.x = 3$, and $\exists (2, 1) \in L_{out}(v_5)$ such that it satisfies Case (1) for $L_{out}(v_{3}) \gg L_{out}(v_{5})$, we have $v_{3} \nrightarrow v_{5}$. Since $\exists (2, 2) \in L_{out}(v_2)$ and $(2, 1) \in L_{out}(v_{5})$ such that it satisfies Case (2) for $L_{out}(v_{2}) \gg L_{out}(v_{5})$, we conclude that $v_2 \nrightarrow v_{5}$.
\end{example}



Algorithm~\ref{alg:query} first uses the chain code of the query vertices, $u$ and $v$, to check whether $u$ and $v$ are in the same chain. If $u$ and $v$ are in the same chain, then by the definition of chain and the fact that $G$ is a DAG, we have $u \rightarrow v$ if $u.y \le v.y$ and $u \nrightarrow v$ if $u.y > v.y$. Then, the algorithm applies the operators $\oplus$ and $\gg$ on the labels of $u$ and $v$ to further examine whether the query answer can be determined. If not, then the algorithm processes the query by testing if any of the descendants of $u$ can reach $v$, by visiting the descendants in a depth-first manner. If a descendant of $u$ can reach $v$, then it implies $u \rightarrow v$. Otherwise, the algorithm finally returns $u \nrightarrow v$. Note that we can prune some descendants of $u$ in the search, which will be discussed in Section~\ref{sec:improvement}.


\begin{theorem} \label{th:query}
Algorithm~\ref{alg:query} correctly answers a reachability query whether $u \rightarrow v$.
\end{theorem}


\begin{proof}
The correctness follows from Lemmas~\ref{le:oplus} and~\ref{le:gg}, and also the fact that $G$ is a DAG.
\end{proof}



\subsection{Time-Based Queries}   \label{ssec:timequery}

We now discuss how to answer temporal reachability queries and minimum temporal path queries.


\noindent \textbf{Temporal reachability queries. } To answer a reachability query whether a source vertex $a$ can reach a target vertex $b$ in a temporal graph $\mathcal{G}$ within a time interval $[t_{\alpha}, t_{\omega}]$, we process the query in the transformed graph $G$ of $\mathcal{G}$ as follows. 

We first find $\langle a,t_{out} \rangle$ in $V_{out}(a)$, where $t_{out}= \min\{ t : \langle a,t \rangle \in V_{out}(a), t \geq t_{\alpha}\}$. Since the vertices in $V_{out}(a)$ are ordered by their time stamp, we find $\langle a,t_{out} \rangle$ by binary search. Similarly, we find $\langle b,t_{in} \rangle$ in $V_{in}(b)$, where $t_{in}= \max\{ t : \langle b,t \rangle \in V_{in}(b), t \leq t_{\omega}\}$.

Let $u=\langle a,t_{out} \rangle$ and $v=\langle b,t_{in} \rangle$. If $u$ or $v$ does not exist, then the answer to the query is false. Otherwise, Algorithm~\ref{alg:query} is called to answer whether $u$ can reach $v$ in $G$. If Algorithm~\ref{alg:query} returns true, then $a$ can reach $b$ in $\mathcal{G}$ within $[t_{\alpha}, t_{\omega}]$. Otherwise, $a$ cannot reach $b$ within $[t_{\alpha}, t_{\omega}]$. 

In addition, if Lines~\ref{ln:bfs0}-\ref{ln:bfs} of Algorithm~\ref{alg:query} need to be executed, we can employ the time interval $[t_{\alpha}, t_{\omega}]$ for search space pruning as follows. For any descendant $w=\langle c,t \rangle$ of $u$ visited during the search, if $t > t_{\omega}$, we can directly terminate the search from $w$.


The above procedure, however, may give an incorrect query answer in the case when $a=b$, i.e., $u.x=v.x$,  $u$ and $v$ are in the same chain. This is because the chains of $G$ obtained in Section~\ref{ssec:chain} may present false reachability information, i.e., $u$ is ordered before $v$ in a chain but $u$ cannot reach $v$ in $G$. However, this can be easily addressed as follows. In the case when $u \in V_{out}(a)$ and $v \in V_{in}(a)$, where $u.x=v.x$, we simply call Algorithm~\ref{alg:query} to answer whether $\exists w \in W$, where $W=\{w: w \in \Gamma_{out}(u',G), u'.x=u.x, u'.y \ge u.y, w.x \ne u.x\}$ (i.e., $u'$ is $u$ or any vertex ordered after $u$ in the same chain, and $w$ is an out-neighbor of $u'$ that is not in the same chain of $u'$), such that $w \rightarrow v$. We have $u \rightarrow v$ if and only if $w$ exists.

The following theorem proves the correctness of processing a temporal reachability query.


\begin{theorem} \label{th:query2}
The algorithm described above correctly answers a reachability query whether $a$ can reach $b$ in $\mathcal{G}$ within $[t_{\alpha}, t_{\omega}]$.
\end{theorem}
\begin{proof}
Since Theorem~\ref{th:mergechain} proves that the labels are constructed correctly given the chain cover defined on $G_{new}$, we examine whether false information presented in $G_{new}$ may lead to a wrong query answer. As shown in Section~\ref{ssec:chain}, the false information is ``$u_i \rightarrow u_j$'' presented in a chain $C$ but $u_i$ cannot reach $u_j$ in $G$. Thus, if we compare $u_i.y$ with $u_j.y$ in Algorithm~\ref{alg:query}, we obtain a wrong result. We show that such a comparison does not happen.

The proof of Theorem~\ref{th:mergechain} shows that the false information happens only when $u_i \in V_{out}(v)$ and $u_j \in V_{in}(v)$. However, according to the construction of $G$ and the fact that Lines~\ref{ln:bfs0}-\ref{ln:bfs} of Algorithm~\ref{alg:query} traverse $G$ instead of $G_{new}$, the case that $u_i \in V_{out}(v)$ and $u_j \in V_{in}(v)$, where $u_i.x = u_j.x$, occurs only when $(u_i,u_j)$ is the input query of Algorithm~\ref{alg:query}. For such a query, we avoid comparing $u_i.y$ with $u_j.y$, as well as comparing $u_{i'}.y$ with $u_j.y$ for any $u_{i'} \in   V_{out}(v)$ that is ordered after $u_i$ in the same chain, by calling Algorithm~\ref{alg:query} to answer whether $w \rightarrow u_j$ for an out-neighbor $w$ of $u_i$ or $u_{i'}$, where $w$ is not in the same chain as $u_i$. According to the construction of $G$, $u_i$ can reach $u_j$ in $G$ if and only if there exists $w$ such that $w \rightarrow u_j$. Thus, the algorithm correctly answers the query.
\end{proof}





\noindent \textbf{Earliest-arrival time.} We compute the earliest-arrival time going from vertex $a$ to vertex $b$ within $[t_{\alpha}, t_{\omega}]$ as follows. We first find $\langle a,t_{out} \rangle$ in $V_{out}(a)$, where $t_{out}= \min\{ t : \langle a,t \rangle \in V_{out}(a), t \geq t_{\alpha}\}$. Then, we compute the set of vertices, $B=\{\langle b,t \rangle : \langle b,t \rangle \in V_{in}(b), t_{\alpha} \leq t \leq t_{\omega} \}$.

Let $u=\langle a,t_{out} \rangle$. We want to find $v = \langle b,t \rangle \in B$ such that $u \rightarrow v$, where $\nexists v'$$=$$\langle b,t' \rangle \in B$ such that $u \rightarrow v'$ and $t'$$<$$t$. If such a vertex $v$ can be found, then the earliest-arrival time going from $a$ to $b$ by any path in $\mathcal{G}$ within $[t_{\alpha}, t_{\omega}]$ is given by $t$. If $v$ is not found, then the corresponding earliest-arrival path does not exist in $\mathcal{G}$.

As vertices in $B$ are ordered according to their time stamp, we can employ a binary-search-like process to find $v$, instead of querying whether $u \rightarrow w$ for each $w \in B$. Let $B = \{w_1, \ldots, w_h\}$ and $w_i = \langle b,t_i \rangle$, where $t_i < t_{i+1}$ for $1 \le i < h$. We start with $w_h$. If $u \nrightarrow w_h$, then $u \nrightarrow w_i$ for $1 \le i \le h$; thus, we can conclude that the earliest-arrival path from $a$ to $b$ does not exist in $\mathcal{G}$ within $[t_{\alpha}, t_{\omega}]$. If $u \rightarrow w_h$, then we choose the middle vertex in $B$, i.e., $w_{h/2}$, and process the query whether $u \rightarrow w_{h/2}$. In this way, we stop until we find the first vertex $w_i \in B$ where $u \rightarrow w_i$, and return $t_i$ as the query answer.

We process each query $u \rightarrow w_i$ by Algorithm~\ref{alg:query}. The correctness of the query answer follows from the fact that an earliest-arrival path from $a$ to $b$ is simply a path starting from $a$ that \emph{reaches} $b$ at the earliest time.



\noindent \textbf{Minimum duration.} We compute the minimum duration taken to go from vertex $a$ to vertex $b$ within $[t_{\alpha}, t_{\omega}]$ as follows. We first compute $A=\{\langle a,t \rangle : \langle a,t \rangle \in V_{out}(a), t_{\alpha} \leq t \leq t_{\omega} \}$. Then, from each $u_i=\langle a,t_i \rangle \in A$, we obtain a starting time $t_i$, and find the earliest-arrival time going from $a$ to $b$ within $[t_i, t_{\omega}]$ by the same binary-search-like process discussed above for computing earliest-arrival time. Let $t'_i$ be the earliest-arrival time obtained starting at time $t_i$. Then, the minimum duration is given by $\min\{(t'_i - t_i): u_i=\langle a,t_i \rangle \in A\}$. The correctness of the query answer follows from the fact that a fastest path from $a$ (starting at time $t$) to $b$ is also an earliest-arrival path from $a$ (starting at time $t$) to $b$.

\section{Improvements on Labeling}  \label{sec:improvement}

We present two improvements on our labeling scheme.



\noindent \textbf{Label reduction. } \ With a close investigation of the property of the transformed graph, we can reduce the label size by half as follows.

Given a vertex $\langle a,t_{out} \rangle \in V_{out}(a)$, let $\langle a,t_{in} \rangle \in V_{in}(a)$ where $t_{in}= \max\{ t : \langle a,t \rangle \in V_{in}(a), t \leq t_{out} \}$. Let $u= \langle a,t_{out} \rangle$ and $v=\langle a,t_{in} \rangle$. Then, we only need to keep $L_{out}(u)$ for $u$, and keep a pointer to $L_{in}(v)$. When $L_{in}(u)$ is needed for query processing, we simply use $L_{in}(v)$ instead.

Similarly, given a vertex $\langle a,t_{in} \rangle \in V_{in}(a)$, let $\langle a,t_{out} \rangle \in V_{out}(a)$ where $t_{out}= \min\{ t : \langle a,t \rangle \in V_{out}(a), t \geq t_{in} \}$. Let $u= \langle a,t_{in} \rangle$ and $v=\langle a,t_{out} \rangle$. We only need to keep $L_{in}(u)$ for $u$, and keep a pointer to $L_{out}(v)$. When $L_{out}(u)$ is needed for query processing, we simply use $L_{out}(v)$ instead.

The following lemma shows the correctness of label reduction.


\begin{lemma} \label{le:reduction}
Label reduction does not affect the correctness of processing a temporal reachability query.
\end{lemma}



\begin{proof}
We first consider answering a query whether $a$ can reach $b$ in a temporal graph $\mathcal{G}$ within a time interval $[t_{\alpha}, t_{\omega}]$, where $a \ne b$. We transform the query to a query in $G$ as discussed in Section~\ref{ssec:timequery}, and let $u=\langle a,t_{out} \rangle \in V_{out}(a)$ and $v=\langle b,t_{in} \rangle \in V_{in}(b)$ be the two corresponding query vertices in $G$. Let $u'=\langle a,t_{u'} \rangle \in V_{in}(a)$ where $t_{u'}= \max\{ t : \langle a,t \rangle \in V_{in}(a), t \leq t_{out} \}$, and $v'=\langle b,t_{v'} \rangle \in V_{out}(b)$ where $t_{v'}= \min\{ t : \langle b,t \rangle \in V_{out}(b), t \geq t_{in} \}$. We show that using $L_{in}(u')$ instead of $L_{in}(u)$ and $L_{out}(v')$ instead of $L_{out}(v)$ will not affect the correctness.

Suppose $u \rightarrow v$. According to the construction of $G$, $u' \rightarrow u$ and $v \rightarrow v'$. Then, $u' \rightarrow u \rightarrow v \rightarrow v'$. In Algorithm~\ref{alg:query}, to answer whether $u \rightarrow v$, $L_{in}(u)$ is only used to check whether $L_{in}(v) \gg L_{in}(u)$. Note that $L_{in}(u)$ is not an operand of $\oplus$. Since $u' \rightarrow v$, $L_{in}(v) \gg L_{in}(u')$ is not true and hence does not report $u \nrightarrow v$. Similarly, since $u \rightarrow v'$, $L_{out}(u) \gg L_{out}(v')$ is not true.

Now suppose $u \nrightarrow v$. According to the construction of $G$, $u' \nrightarrow v$ and $u \nrightarrow v'$. Thus, using $L_{in}(u')$ instead of $L_{in}(u)$ and $L_{out}(v')$ instead of $L_{out}(v)$ will not give $u \rightarrow v$.

Finally, for the case $a = b$, we have $u.x=v.x$. This case is resolved by transforming the query $(u,v)$ into other queries in the form of $(w,v)$ where $w.x \ne v.x$, as discussed in Section~\ref{ssec:timequery}.
\end{proof}








\noindent \textbf{Topological-sort-based labels.} We can further prune unreachable vertices to reduce the querying cost by some \emph{light-weight} labels.


%

We first introduce the \emph{topological level number}~\cite{ChengHWF13sigmod,SeufertABW13icde,WeiYLJ14pvldb,YildirimCZ12vldb} for a vertex $v$, denoted by \ $\ell(v)$:
\begin{itemize}
  \item If \ $\Gamma_{in}(v,G) = \emptyset$$:$ \ $\ell(v)=1$;
  \item Else$:$ \ $\ell(v)=\max\{(\ell(u)+1): u \in \Gamma_{in}(v,G)\}$.
\end{itemize}

We use $\ell(v)$ in processing a reachability query as follows. If $\ell(u) \ge \ell(v)$ and $u \ne v$, then $u \nrightarrow v$. This is true because if $u \rightarrow v$, then $v$ is a descendant of $u$ and hence $\ell(u) < \ell(v)$. We can compute $\ell(v)$ for each $v \in V$ in linear time using a single topological sort of $G$.


A topological sort also gives an ordering of the vertices in $V$. Let $\sigma(v)$ be the position of a vertex $v$ in a topological ordering of $V$, where a vertex is ordered before its out-neighbors. We can use $\sigma(v)$ in processing a reachability query as follows. If $\sigma(u) > \sigma(v)$, then $u \nrightarrow v$. 
Note that topological ordering of $V$ may not be unique, and this can be employed to increase the pruning power. We compute topological sort by DFS, and generate two topological orderings of $V$ by visiting the out-neighbors of a vertex $v$ according to their original order in $\Gamma_{out}(v,G)$ as well as their reverse order in $\Gamma_{out}(v,G)$. Let $\sigma_1(v)$ and $\sigma_2(v)$ be the value of $\sigma(v)$ obtained from the two topological orderings of $V$. If either $\sigma_1(u) > \sigma_1(v)$ or $\sigma_2(u) > \sigma_2(v)$, then $u \nrightarrow v$.


%
%


\section{Performance Evaluation}   \label{sec:result}


We now report the performance of TopChain. We ran all the experiments on a machine with an Intel 2.0GHz CPU and 128GB RAM, running Linux.

\vspace{2mm}

\noindent \textbf{Datasets.} \ We use 15 real temporal graphs, 6 of them, {\tt austin}, {\tt berlin}, {\tt houston}, {\tt madrid}, {\tt roma} and {\tt toronto}, are from Google Transit Data Feed project (code.google.com/p/googletransitdatafeed/wiki/PublicFeeds), where each dataset represents the public transportation network of a city. The other 9 of them are from the Koblenz Large Network Collection (konect.uni-koblenz.de/), and we selected one large temporal graph from each of the following categories: {\tt amazon-ratings} ({\tt amazon}) from the Amazon online shopping website; \ {\tt arxiv-HepPh} ({\tt arxiv}) from the arxiv networks; \ {\tt dblp-coauthor} ({\tt dblp}) from the DBLP coauthor networks; \ {\tt delicious-ut} ({\tt delicious}) from the network of ``delicious''; \ {\tt enron} from the email networks; \ {\tt flickr-growth} ({\tt flickr}) from the social network of Flickr; \ {\tt wikiconflict}   ({\tt wikiconf}) indicating the conflicts between users of Wikipedia; \ {\tt wikipedia-growth}  ({\tt wikipedia})  from the English Wikipedia hyperlink network; \ {\tt youtube} from the social media networks of YouTube.


Table \ref{tab:realdata} gives some statistics of the datasets. We show the number of vertices and edges in each temporal graph $\mathcal{G}=(\mathcal{V},\mathcal{E})$ and the transformed graph $G=(V,E)$ of $\mathcal{G}$. The value of $\pi$ varies significantly for different graphs, indicating the different levels of temporal activity between two vertices in each $\mathcal{G}$. We also show the number of atomic time intervals in each $\mathcal{G}$, denoted by $|T_\mathcal{G}|$. If we break $\mathcal{G}$ into snapshots by atomic time intervals, the {\tt wikiconf} graph consists of as many as 273,909 snapshots.

\begin{table}[!tbp]
\caption{Datasets} \label{tab:realdata}
\begin{center}
\small
\resizebox{\linewidth}{!}{
\begin{tabular}{|l||r|r|r|r|r|r|}
\hline
Dataset & $|\mathcal{V}|$  & $|\mathcal{E}|$ & $\pi$ & $|T_\mathcal{G}|$  & $|V|$ & $|E|$\\
\hline \hline
{\tt austin} & 2,676 & 320,652	&659	&100,928	&629,664	&1,253,961 \\
\hline
{\tt berlin} & 12,845& 2,093,977 &2,221	&109,500	&3,175,993&6,753,520 \\
\hline
{\tt houston} & 9,848& 1,123,580	&783	&98,820	&2,205,384	&4,396,434 \\
\hline
{\tt madrid} & 4,636 &1,917,090&2,406&110,347&3,793,545	&7,590,572 \\
\hline
{\tt roma} & 8,779& 2,290,762	&2,170	&109,392&	4,431,239	&8,881,221 \\
\hline
{\tt toronto} &10,790&3,310,871&1,664&	109,660	&6,415,493&12,875,896 \\
\hline
{\tt amazon} & 2,146,057 & 5,776,660 & 28 & 3,329 & 9,883,393 & 13,166,635 \\
\hline
{\tt arxiv} & 28,093  & 9,193,606 & 262 & 2,337 & 433,412 & 9,759,445 \\
\hline
{\tt dblp} & 1,103,412  & 11,957,392 & 38 & 70 & 5,553,200 & 16,976,956 \\
\hline
{\tt delicious} & 4,535,197 & 219,581,041 & 1,070 & 1,583 & 73,792,065 & 293,632,816 \\
\hline
{\tt enron} & 87,273  & 1,134,990 & 1,074 & 213,218 & 1,366,786 & 2,504,928 \\
\hline
{\tt flickr} & 2,302,925  & 33,140,017 & 1 & 134 & 12,600,099 & 44,358,410 \\
\hline
{\tt wikiconf} & 118,100  & 2,917,777 & 562 & 273,909 & 3,191,271 & 6,009,300 \\
\hline
{\tt wikipedia} & 1,870,709  & 39,953,145 & 1 & 2,198 & 34,814,941 & 77,196,220 \\
\hline
{\tt youtube} & 3,223,589 & 12,223,774 & 2 & 203 & 11,497,869 & 21,139,520 \\
\hline
\end{tabular}
}
\end{center}
\end{table}


\subsection{Performance on Reachability Queries}  \label{result:reachability}

Existing reachability indexes can be categorized into three groups: (1)\emph{Transitive Closure}, (2)\emph{2-Hop Labels}, and (3)\emph{Label+Search}. We compare with the state-of-the-art indexes in each category: \textbf{PWAH8}~\cite{SchaikM11sigmod} in (1); \textbf{TOL}~\cite{ZhuLWX14sigmod} in (2); and \textbf{GRAIL++}~\cite{YildirimCZ12vldb}, \textbf{Ferrari}~\cite{SeufertABW13icde} and \textbf{IP+}~\cite{WeiYLJ14pvldb} in (3). We obtained the source codes from the authors. All the source codes are in C++, and we compiled them and TopChain using the same g++ compiler and optimization option.


We report the index size, indexing time, and querying time in Table~\ref{tab:indexSize}, Table~\ref{tab:indexTime}, and Table~\ref{tab:queryTime}, respectively (note that TTL is to be discussed in Section~\ref{result:path}; TC1 and TC2 are variants of TopChain to be discussed in Section~\ref{result:topchain}). The best results are highlighted in \textbf{bold}. The sign ``-'' in the tables indicates that PWAH8 or TOL or TTL cannot be constructed within time $\max(x,y)$, where $x$ is 100 times of the indexing time of TopChain and $y$ is 10,000 seconds. For example, PWAH8, TOL and TTL cannot be constructed in $x=38,448$ seconds for {\tt delicious}.




We set $k$ to 5 for all the \emph{Label+Search} indexes, i.e., TopChain, IP+, Ferrari and GRAIL++. The effect of $k$ will be studied in Section~\ref{result:topchain}, in general, query performance improves if we use a larger $k$, but a larger $k$ also leads to a larger index and longer indexing time. Since $k$ sets the number of labels for each vertex, with the same $k$ value, the four \emph{Label+Search} indexes have comparable sizes, as shown in Table~\ref{tab:indexSize}. In comparison, the index sizes of PWAH8, TOL and TTL are much larger. This is because that the index size of TopChain, IP+, Ferrari and GRAIL++ are linear to the graph size, while PWAH8, TOL and TTL may take quadratic space of the graph size.


For indexing efficiency, Table~\ref{tab:indexTime} shows that TopChain is the fastest in 11 out of 15 datasets, while for the other 4 datasets, the indexing time of TopChain is close to the best one. Compared with PWAH8 and TOL, TopChain is clearly much more scalable. PWAH8 and TOL have much worse performance than TopChain in terms of both index size and indexing time.


\begin{table}[!tbp]
\caption{Index size (in MB)} \label{tab:indexSize}
\begin{center}
\small
\resizebox{\linewidth}{!}{
\begin{tabular}{|l||r|r|r|r|r|r|r|}
\hline
Dataset & TopChain & IP+ & Ferrari & GRAIL++ & PWAH8 & TOL & TTL \\
\hline \hline
{\tt austin} & 31& \textbf{29}& 34& 38 & 296& 497& 110\\
\hline
{\tt berlin} & 157& \textbf{145}& 173& 354 & 3793& 1443& 439\\
\hline
{\tt houston}& 109& \textbf{101}& 120& 135 & 1993& 2129& 503\\
\hline
{\tt madrid} & 188& \textbf{174}& 206& 232& 8997& 6907& 870\\
\hline
{\tt roma} & 219& \textbf{203}& 241& 270 & 11097& 6075& 1062\\
\hline
{\tt toronto} & 318& \textbf{294}& 349& 392 & 11646& 5347& 905\\
\hline
{\tt amazon} & \textbf{373} & 395 & 482 & 603 & 31886 & 2532 & - \\
\hline
{\tt arxiv} & 21 & \textbf{19} & 24 & 26 & 53 & 6921 & 306\\
\hline
{\tt dblp} & 235 & \textbf{216} & 256 & 339 & 47390 & - & -\\
\hline
{\tt delicious} & 3410 & \textbf{3293} & 4010 & 4504 & - & - & - \\
\hline
{\tt enron} & \textbf{60} & 61 & 72 & 83 & 321 & 312 & 94\\
\hline
{\tt flickr} & 524 & \textbf{491} & 558 & 769 & - & - & -\\
\hline
{\tt wikiconf} & \textbf{124} & 142 & 176 & 195 & 1107 & 1105 & 214\\
\hline
{\tt wikipedia} & 1631 & \textbf{1542} & 1948 & 2125 & - & - & -\\
\hline
{\tt youtube} & 464 & \textbf{441} & 496 & 702 & - & - & - \\
\hline
\end{tabular}
}

\end{center}
\end{table}

\begin{table}[!tbp]
\caption{Indexing time (in seconds)} \label{tab:indexTime}
\begin{center}
\small
\resizebox{\linewidth}{!}{
\begin{tabular}{|l||r|r|r|r|r|r|r|}
\hline
Dataset & TopChain & IP+ & Ferrari & GRAIL++ & PWAH8 & TOL & TTL\\
\hline \hline
{\tt austin}& \textbf{0.98}& 1.14& 2.88& 2.71& 38.78& 79.39& 50.57\\
\hline
{\tt berlin}& \textbf{5.34}& 6.06& 13.59& 13.99& 493.97& 350.39& 566.33\\
\hline
{\tt houston}& \textbf{3.91}& 4.05& 9.58& 10.72& 247.98& 368.61& 306.20\\
\hline
{\tt madrid}& \textbf{6.14}& 7.10& 17.56& 18.07& 1158.39& 3078.81& 1702.01\\
\hline
{\tt roma}&\textbf{7.40}& 8.51& 19.17& 21.06& 1376.84& 1574.39& 2062.25\\
\hline
{\tt toronto}& \textbf{11.63}& 13.11& 29.21& 30.72& 1405.75& 1173.33& 1061.96\\
\hline
{\tt amazon} & \textbf{28.03} & 26.54 & 43.35 & 72.02 & 2495.06 & 817.75 & -\\
\hline
{\tt arxiv} & \textbf{2.37} & 4.78 & 4.80 & 8.08 & 27.73 & 6730.17 & 2761.00\\
\hline
{\tt dblp} & 17.82 & \textbf{17.31} & 31.40 & 45.91 & 7260.61 & - & -\\
\hline
{\tt delicious} & \textbf{384.48} & 489.63 & 848.24 & 778.89 & - & - & -\\
\hline
{\tt enron} & 2.37 & \textbf{2.26} & 4.57 & 6.07 & 33.44 & 64.79 & 603.73\\
\hline
{\tt flickr} & 55.69 & \textbf{53.82} & 83.57 & 139.10 & - & - & -\\
\hline
{\tt wikiconf} & \textbf{5.50} & 6.11 & 11.95 & 15.65 & 111.96 & 246.12 & 2530.21\\
\hline
{\tt wikipedia} & 151.29 &\textbf{ 142.57} & 259.06 & 298.07 & - & - & -\\
\hline
{\tt youtube} & \textbf{30.82} & 33.43 & 54.04 & 78.32 & - & - & - \\
\hline

\end{tabular}
}

\end{center}
\end{table}


\begin{table}[!tbp]
\caption{Total querying time (in milliseconds)} \label{tab:queryTime}
\begin{center}
\small
\resizebox{1.02\linewidth}{!}{
\begin{tabular}{|l||r|r|r|r|r|r|r|r|}
\hline
Dataset & TopChain & TC1 & TC2 & IP+ & Ferrari & GRAIL++ & PWAH8 & TOL \\
\hline \hline
{\tt austin}& 1.07 & 2.23& 3.84& 465.04& 282.21& 437.00& \textbf{1.04}& 4.58\\
\hline
{\tt berlin}& \textbf{1.23}& 4.55& 5.31& 10867.10& 5520.82& 4063.79& 1.64& 2.74\\
\hline
{\tt houston}& \textbf{0.54}& 0.56& 0.91& 3325.21& 2837.12& 1619.67& 1.20& 3.53\\
\hline
{\tt madrid}& \textbf{0.52}& 0.53& 1.46& 3198.42& 2974.85& 1971.84& 1.15& 3.17\\
\hline
{\tt roma}& \textbf{0.69}& 0.85& 1.21& 7806.06& 4326.89& 4066.42& 1.52& 3.32\\
\hline
{\tt toronto}& 4.81& 7.87& 11.55& 13899.10& 7397.83& 4631.84& \textbf{1.53}& 3.76\\
\hline
{\tt amazon} & 7.38 & 34.36 & 29.24 & 28.29 & 65.13 & 44.29 & 39.08 & \textbf{5.22} \\
\hline
{\tt arxiv} & \textbf{3.33} & 14.26 & 46.83 & 425.46 & 65.65 & 1154.70 & 23.77 & 4.39 \\
\hline
{\tt dblp} & \textbf{88.12} & 981.26 & 1777.74 & 4353.38 & 3335.07 & 2644.30 & 161.03 & - \\
\hline
{\tt delicious} & \textbf{27.05} & 1246.72 & 2648.84 & 60544.40 & 39117.00 & 4002.09 & - & - \\
\hline
{\tt enron} & \textbf{2.36} & 10.84 & 47.88 & 147.39 & 14.79 & 113.46 & 24.13 & 2.81 \\
\hline
{\tt flickr} & \textbf{21.46} & 91.52 & 254.40 & 4073.06 & 2160.87 & 1495.67 & - & - \\
\hline
{\tt wikiconf} & 7.23 & 24.02 & 176.65 & 701.14 & 51.65 & 406.60 & 38.26 & \textbf{3.16} \\
\hline
{\tt wikipedia} & \textbf{288.50} & 4074.23 & 15235.40 & 44810.10 & 8891.74 & 15758.72 & - & - \\
\hline
{\tt youtube} & \textbf{23.29} & 528.81 & 165.51 & 974.43 & 1099.03 & 324.17 & - & - \\
\hline
\end{tabular}

}
\end{center}
\end{table}

For query processing, we randomly generated 1000 queries, and set $[t_{\alpha}, t_{\omega}]$ to be $[0, \infty]$ for all queries so that query processing accesses the whole transformed graph. For all the indexes tested, we applied the same procedure of processing temporal reachability queries described in Section~\ref{ssec:timequery}.


Table~\ref{tab:queryTime} reports the total querying time by using each index. TopChain is the fastest in 11 out of 15 datasets. Among the four \emph{Label+Search} indexes, TopChain is the fastest in all cases and is from a few times to over two orders of magnitude faster than IP+, Ferrari and GRAIL++. It is particularly important for handling the larger datasets such as {\tt delicious} and {\tt wikipedia}, while other methods have long querying time, TopChain remains to be very efficient. This demonstrates the effectiveness and better scalability of TopChain's labeling scheme for querying reachability in temporal graphs.



\begin{table}[!tbp]
\caption{Total querying time of TopChain, TTL and 1-pass (in seconds)} \label{tab:path}
\begin{center}
\small
\resizebox{1.0\linewidth}{!}{
\begin{tabular}{|l||r|r|r|r|r|r|}
\hline
 & \multicolumn{3}{|c|}{Earliest-arrival} & \multicolumn{3}{|c|}{Fastest}  \\
 \cline{2-7}
& TopChain & TTL & 1-pass & TopChain & TTL &1-pass \\
\hline \hline
{\tt austin}& \textbf{0.006}& 0.016& 0.888& \textbf{0.023}& 0.035& 6.492\\
\hline
{\tt berlin}& \textbf{0.009}& 0.021& 5.192& \textbf{0.012}& 0.032& 9.057\\
\hline
{\tt houston}& \textbf{0.015}& 0.025& 2.542& \textbf{0.015}& 0.046& 12.782\\
\hline
{\tt madrid}& \textbf{0.003}& 0.059& 4.551& \textbf{0.019}& 0.136& 17.029\\
\hline
{\tt roma}& \textbf{0.047}& 0.060& 5.679& \textbf{0.036}& 0.117& 18.005\\
\hline
{\tt toronto}& 0.058& \textbf{0.036}& 8.472& 0.406& \textbf{0.061}& 27.576\\
\hline
{\tt amazon} & \textbf{0.011} & - & 88.440 & \textbf{0.048} & -& 206.499 \\
\hline
{\tt arxiv} & \textbf{0.006} & 0.086 &16.502 & \textbf{0.027} & 0.086 &276.874 \\
\hline
{\tt dblp} & \textbf{1.095} & - & 39.427 & \textbf{1.673} & - & 310.934 \\
\hline
{\tt delicious} & \textbf{0.010} & - & 611.542 & \textbf{0.153} & - & 5471.470 \\
\hline
{\tt enron} & \textbf{0.001} & 0.004 & 3.325 & \textbf{0.002} &\textbf{0.002}& 14.777 \\
\hline
{\tt flickr} & \textbf{0.109} & - & 125.018 & \textbf{0.469}& - & 1007.670 \\
\hline
{\tt wikiconf} & 0.018 & \textbf{0.012} &8.201 & 0.066 &\textbf{0.009}& 43.992 \\
\hline
{\tt wikipedia} & \textbf{0.638} & - & 222.182 & \textbf{1.507} & -& 3766.520 \\
\hline
{\tt youtube} & \textbf{0.076} & - & 60.800 & \textbf{0.127}& -& 233.982 \\
\hline

\end{tabular}
}
\end{center}
\end{table}

\subsection{Performance on Time-based Path Queries}  \label{result:path}

We compare TopChain with \textbf{1-pass}\cite{WuCHKLX14pvldb}, and the state-of-the-art indexing method, \textbf{TTL}\cite{WangLYXZ15sigmod}, for computing earliest-arrival time and the duration of a fastest path from a vertex $u$ to another vertex $v$ within time interval $[t_{\alpha}, t_{\omega}]$.

%

Table~\ref{tab:path} reports the total querying time of 1000 queries. TopChain is orders of magnitude faster than 1-pass. TTL is only able to handle 9 small datasets. The indexing time of TTL is orders of magnitude longer than that of TopChain, as reported in Table~\ref{tab:indexTime}; while TTL also has a larger index size than TopChain, as shown in Table~\ref{tab:indexSize}. Given such significantly higher indexing cost, TTL is still only faster than TopChain in just 2 of the 9 datasets that TTL can handle for querying earliest-arrival time, and in just 2 out of the 9 datasets for querying the duration of a fastest path. The result thus demonstrates the efficiency of our method for answering queries in a temporal graph in real time, while it also shows that TopChain is more scalable than existing methods for processing large temporal graphs.

\begin{figure}[!tbp]
\begin{minipage}[b]{0.49\linewidth}
\centering
\includegraphics[width=\textwidth]{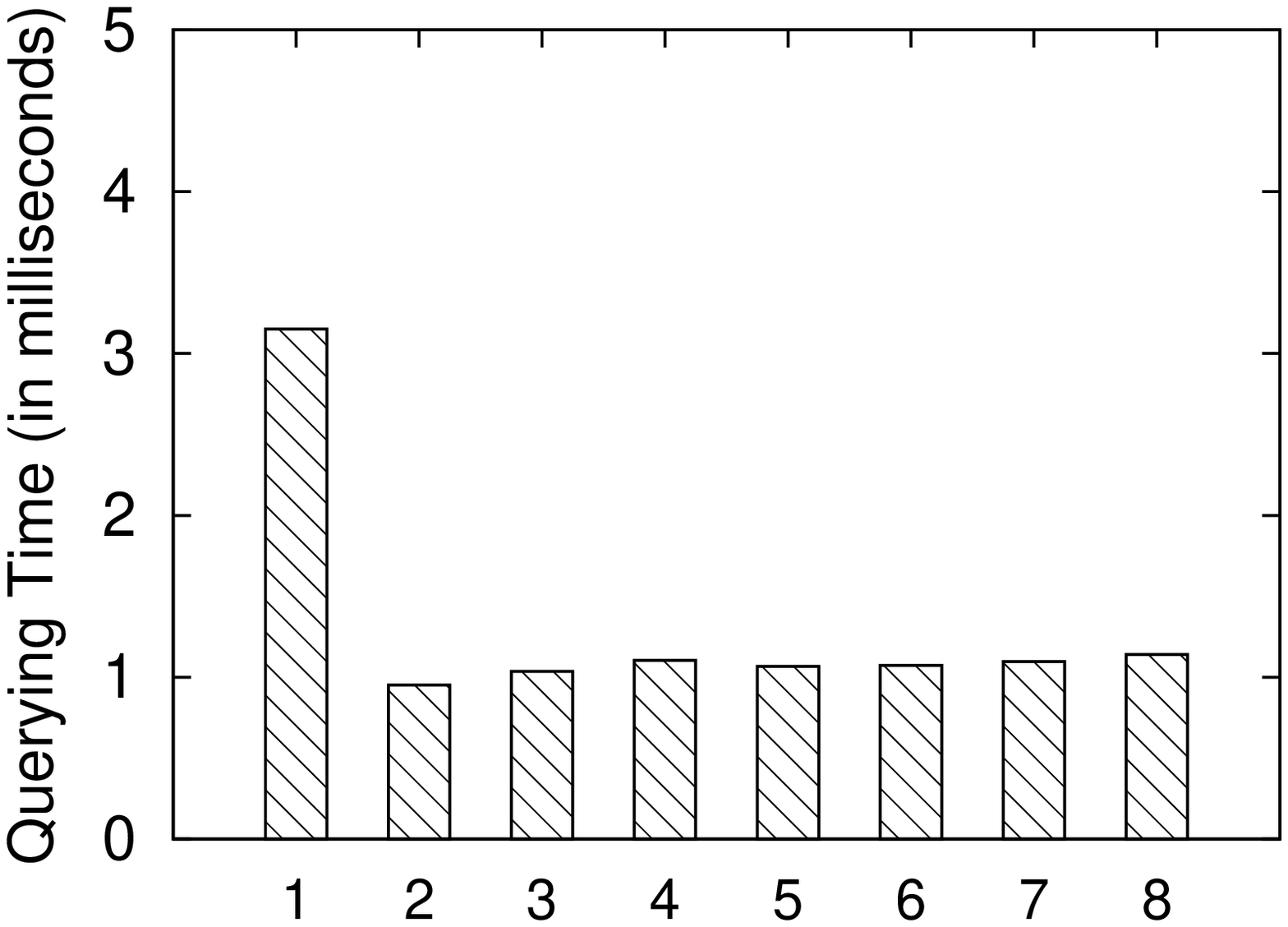}
\caption{Querying time of {\tt austin}}
\label{fig:qtime_k1}
\end{minipage}
\begin{minipage}[b]{0.49\linewidth}
\centering
\includegraphics[width=\textwidth]{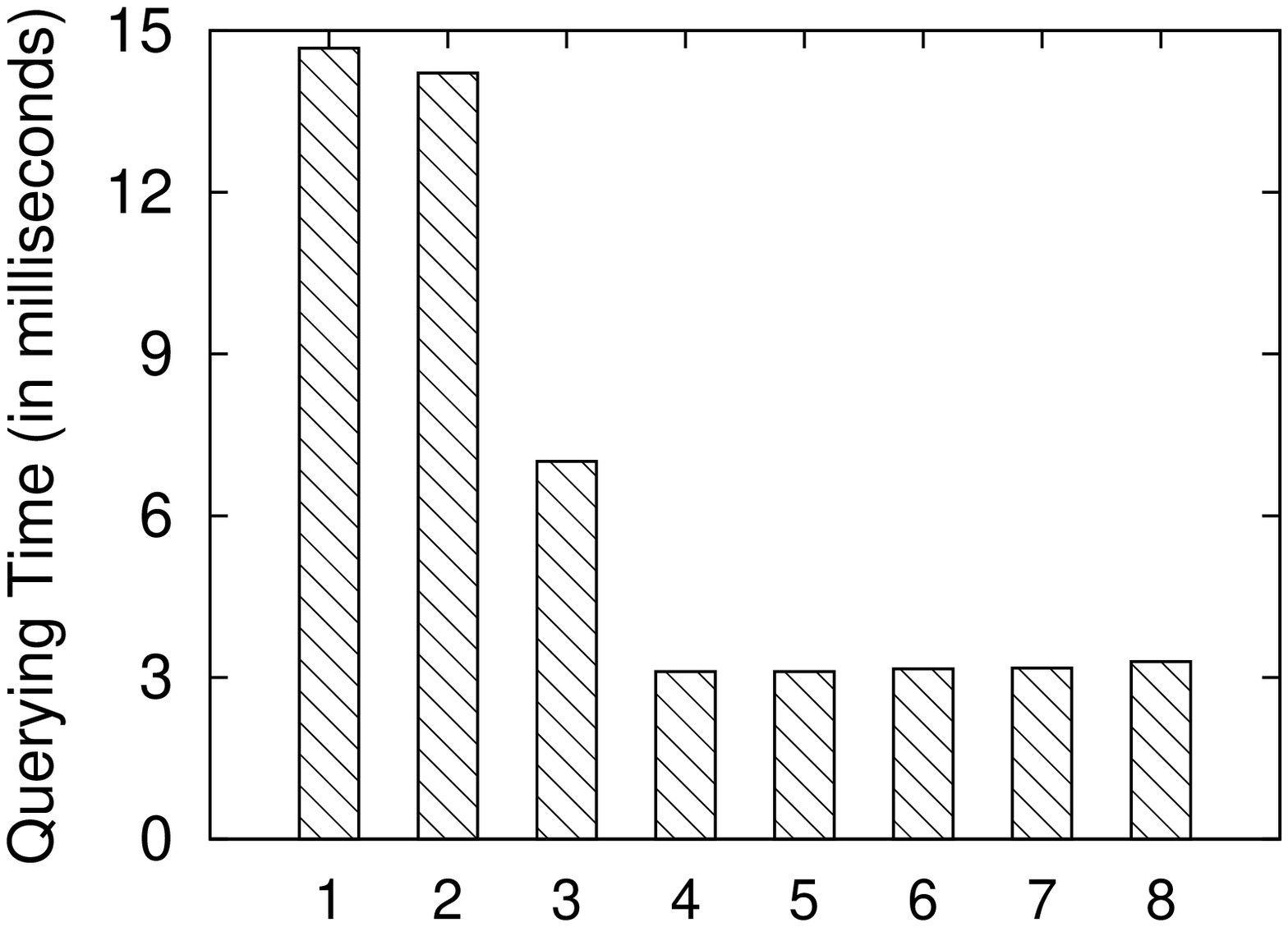}
\caption{Querying time of {\tt arxiv}}
\label{fig:qtime_k2}
\end{minipage}
\end{figure}

\begin{table}[!tbp]
\caption{Total querying time for varying intervals (in milliseconds)} \label{tab:interval}
\begin{center}
\small
\resizebox{0.75\linewidth}{!}{
\begin{tabular}{|l||r|r|r|r|}
\hline
Dataset & \textit{$I_1$} & \textit{$I_2$} & \textit{$I_3$} & \textit{$I_4$} \\
\hline \hline
{\tt austin}& 1.067& 0.310& 0.020& 0.015\\
\hline
{\tt berlin}& 1.230& 0.551& 0.430& 0.044\\
\hline
{\tt houston}& 0.544& 0.051& 0.044& 0.026\\
\hline
{\tt madrid}& 0.516& 0.333& 0.223& 0.060\\
\hline
{\tt roma}& 0.693& 0.367& 0.052& 0.032\\
\hline
{\tt toronto}& 4.812& 1.091& 0.158& 0.076\\
\hline
{\tt amazon} & 7.377 & 0.084 & 0.039 & 0.037 \\
\hline
{\tt arxiv} & 3.332 & 0.024 & 0.011 & 0.009 \\
\hline
{\tt dblp} & 645.109 & 0.085 & 0.037 & 0.035 \\
\hline
{\tt delicious} & 27.048 & 0.339 & 0.073 & 0.067 \\
\hline
{\tt enron} & 2.356 & 0.539 & 0.130 & 0.026 \\
\hline
{\tt flickr} & 21.459 & 19.980 & 18.243 & 0.245 \\
\hline
{\tt wikiconf} & 7.234 & 4.443 & 0.627 & 0.148 \\
\hline
{\tt wikipedia} & 288.502 & 0.213 & 0.218 & 0.173 \\
\hline
{\tt youtube} & 293.286 & 9.769 & 0.204 & 0.061 \\
\hline
\end{tabular}
}
\end{center}
\end{table}

\begin{figure*}[!t]
    \centering
    {\includegraphics[width=1.9\columnwidth]{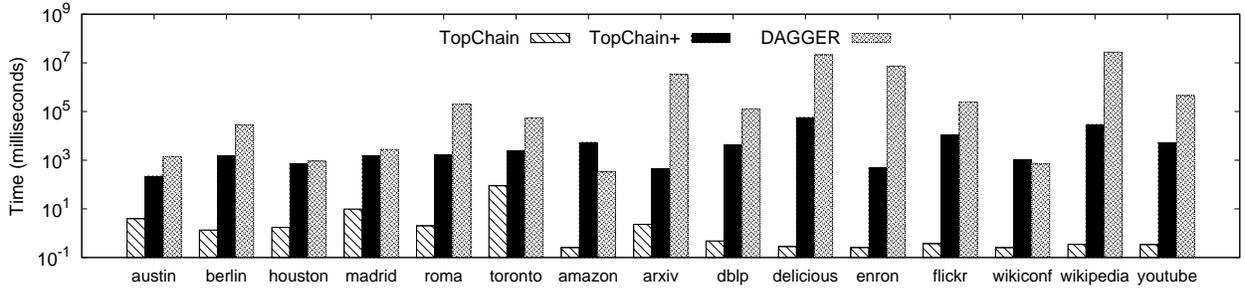}}
    \caption{Average update time due to edge insertion} \label{fig:updateTime}
\end{figure*}

\subsection{Study on TopChain Label}  \label{result:topchain}


Next, we study the strategy of chain cover and chain ranking used in our labeling scheme. As discussed in Section~\ref{ssec:chain}, TopChain merges $V_{in}(v)$ and $V_{out}(v)$ into a chain for each $v$ and rank the chains by degree. Here, we also tested two variants of TopChain: (1)\textbf{TC1}: we compute the chains by the greedy algorithm of~\cite{Simon88tcs} and rank them by degree; and (2)\textbf{TC2}: we merge $V_{in}(v)$ and $V_{out}(v)$ into a chain for each $v$ and rank the chains randomly. The index sizes of TC1 and TC2 are almost the same as TopChain. TopChain and TC2 have similar indexing time, but TC1 takes longer time due to chain computation. For query processing, as reported in Table~\ref{tab:queryTime}, TopChain is significantly faster than both TC1 and TC2. This result demonstrates the effectiveness of our labeling scheme in using the properties of temporal graphs.


We also study the effect of $k$. We report query performance using the two graphs, {\tt austin} and {\tt arxiv}, where  the average degree of the transformed graph of {\tt austin} is around 2 (representing graphs with low degree), while it is larger than 20 for {\tt arxiv} (representing graphs with higher degree). Figures~\ref{fig:qtime_k1} and~\ref{fig:qtime_k2} show that a larger $k$ can improve query performance, but it does not help when $k$ is too large. For dataset {\tt austin}, when $k$ is larger than 2, the querying time does not decrease as $k$ increases. Similarly, for dataset {\tt arxiv}, the query performance does not improve when $k$ is larger than 4. This result shows that a small $k$ value is sufficient for good query performance.

\subsection{Effect of Varying Time Intervals}  \label{result:interval}


The input time interval $[t_{\alpha}, t_{\omega}]$ can affect query performance since a smaller interval gives a smaller search space. We tested four different time intervals, $I_1$ to $I_4$. We set $I_1=[0,|T_\mathcal{G}|]$, where $|T_\mathcal{G}|$ is shown in Table~\ref{tab:realdata}. For each $I_i$, for $1 \leq i \leq 3$, we divide $I_i$ into two equal sub-intervals so that $I_{i+1}$ is the first sub-interval of $I_i$.

We used the same 1000 temporal reachability queries tested in Section~\ref{result:reachability}, but with different input time interval $I_i$. Table~\ref{tab:interval} shows that the total querying time of TopChain decreases significantly from time interval $I_1$ to $I_2$ for most datasets, and then the decrease becomes slowly when the time intervals become smaller. This is because when the time interval becomes smaller, the reachability also drops significantly and thus more queries can be answered directly by TopChain's pruning strategies. But the pruning effect becomes less and less obvious when the time interval is small enough.


\subsection{Performance on Dynamic Updating}  \label{result:update}

We compare TopChain with \textbf{Dagger}~\cite{YildirimCZ13corr}, which is an extension of GRAIL~\cite{YildirimCZ10pvldb} that supports dynamic update. There are other methods that also handle dynamic update in reachability indexing~\cite{BramandiaCN10tkde,DemetrescuI06jcss,RodittyZ04stoc,SchenkelTW05icde,ZhuLWX14sigmod}, but they can handle only a few smaller graphs that we tested and Dagger is the only one that can scale.


Figure~\ref{fig:updateTime} reports the average update time due to edge insertions, where TopChain+ shows the update time including a re-computation of the topological-sort-based labels (presented in Section~\ref{sec:improvement}) for each edge insertion. The result shows that re-computing the topological-sort-based labels dominates the updating time of TopChain, but TopChain is still significantly faster than Dagger. Dagger only performs well when the graph has low average degree, e.g., {\tt amazon}. For query performance, Dagger is worse than GRAIL++~\cite{YildirimCZ12vldb}, which is significantly slower than TopChain as shown in Table~\ref{tab:queryTime}.

\begin{figure*}[!t]
    \centering
    \subfigure[{\small Effect of $|\mathcal{V}|$}]{\includegraphics[width=0.55\columnwidth]{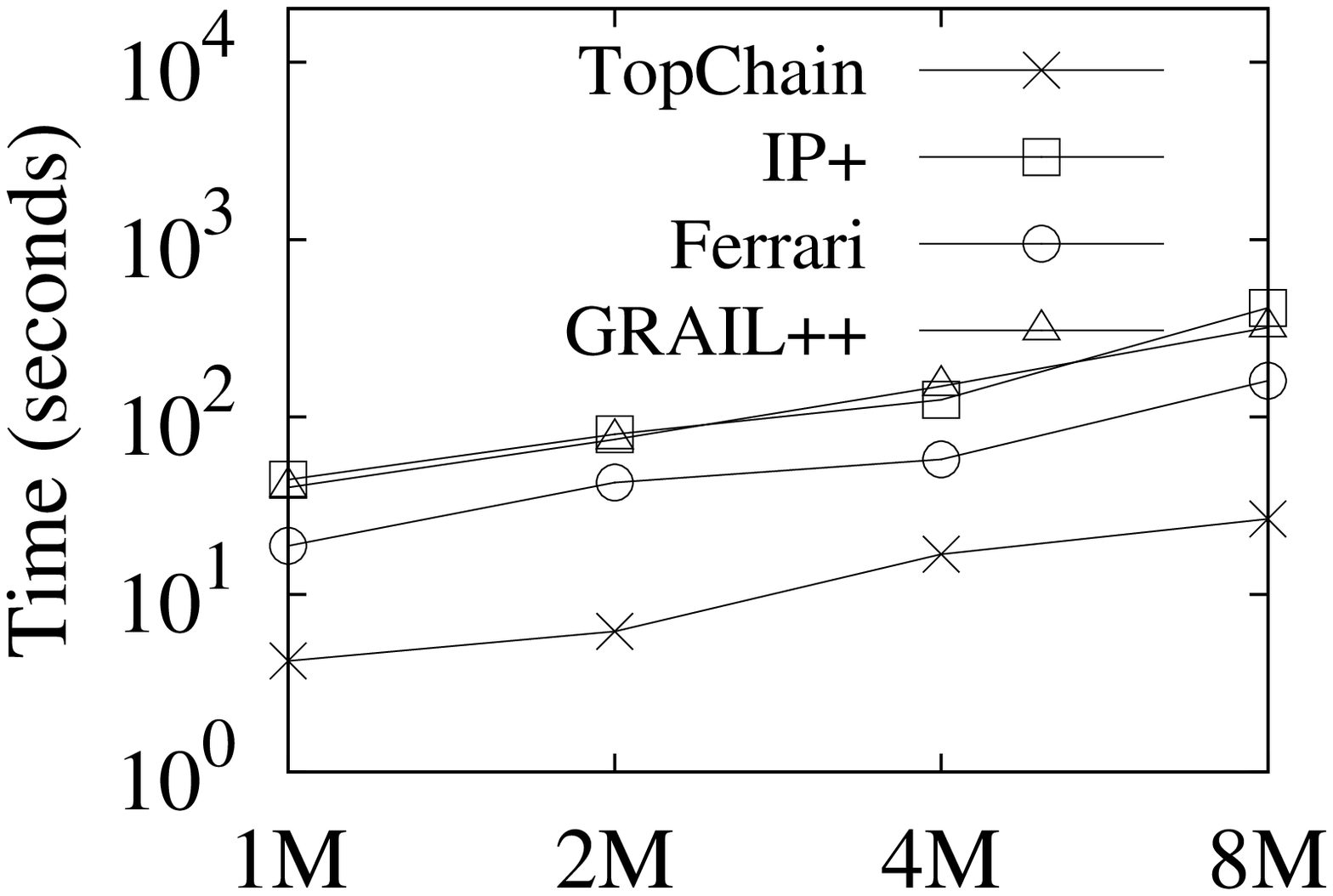}\label{fig:power_v}}
    \subfigure[{\small Effect of $\pi$}]{\includegraphics[width=0.55\columnwidth]{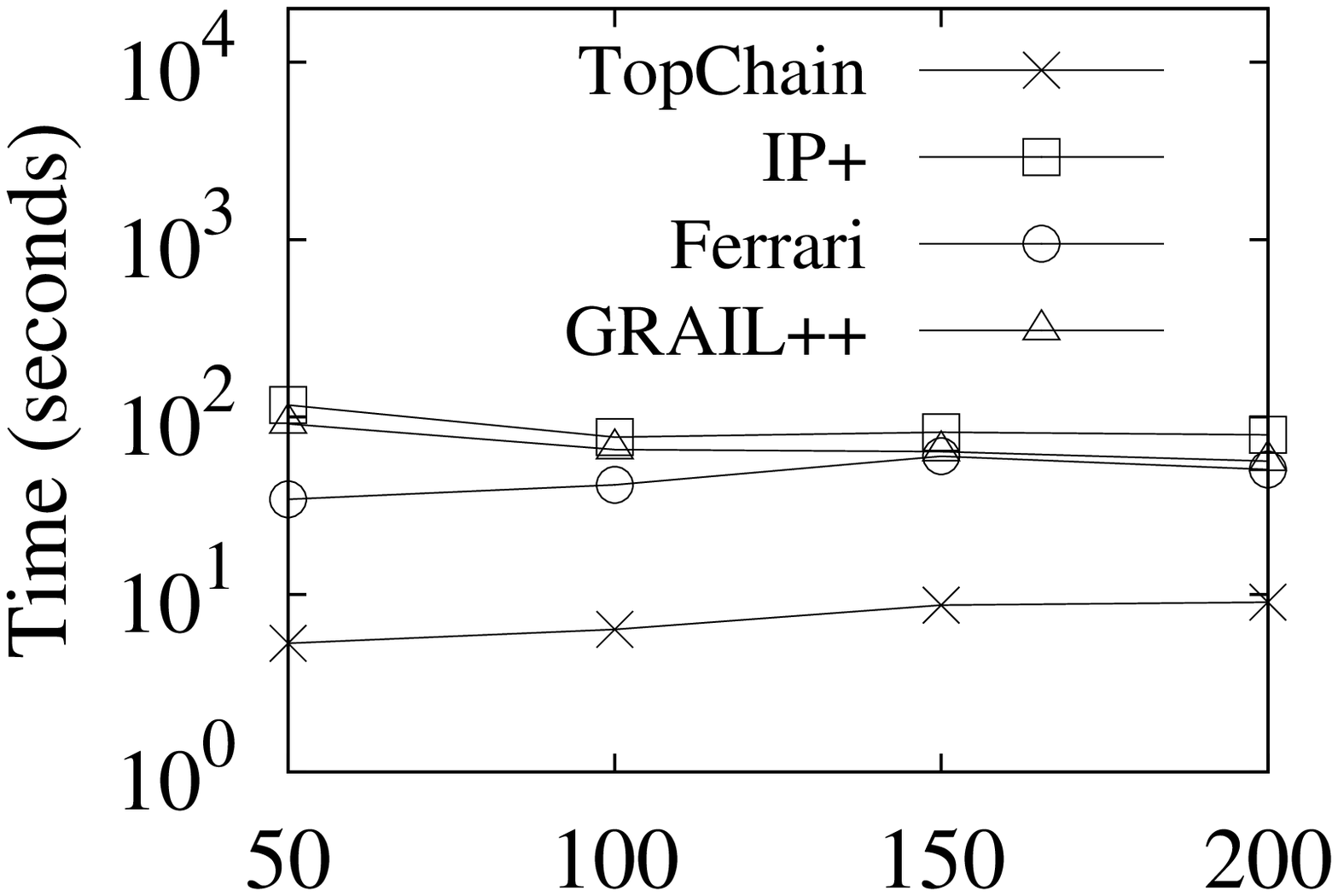}\label{fig:power_pi}}
    \subfigure[{\small Effect of $d_{\it avg}(u,\mathcal{G})$}]{\includegraphics[width=0.55\columnwidth]{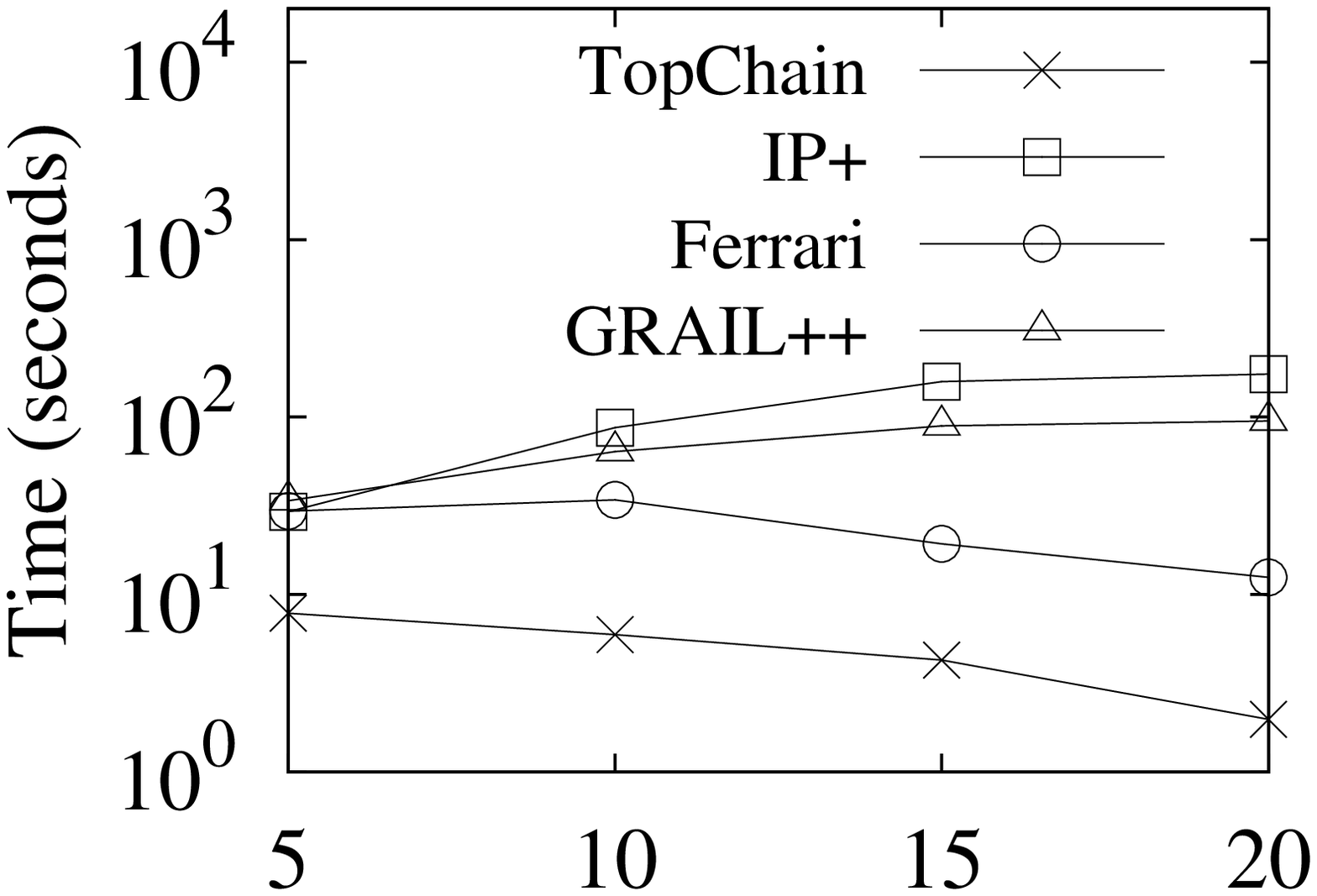}\label{fig:power_td}}
    \caption{Total querying time on synthetic power-law temporal graphs} \label{fig:power}
\end{figure*}


\subsection{Scalability Tests}  \label{result:synthetic}

We generated synthetic temporal graphs for scalability tests. We varied the number of vertices $|\mathcal{V}|$ from 1M to 8M (M $=10^6$), the value of $\pi$ from 50 to 200 which controls the number of multiple temporal edges between two vertices, and the average vertex degree $d_{\it avg}(u,\mathcal{G})$ from 5 to 20. We generated three sets of datasets by varying one of $|\mathcal{V}|$, $\pi$ and $d_{\it avg}(u,\mathcal{G})$, while fixing the other two to their default values: $|\mathcal{V}|$ $=$ 2M, $\pi=100$ and $d_{\it avg}(u,\mathcal{G})=10$.

We compare TopChain with IP+, Ferrari and GRAIL++ since only these methods can scale to large graphs. Since the indexing time and index size of these four methods are comparable, we only report the total querying time for 1000 randomly generated queries.

Figure~\ref{fig:power}(a) shows that all the methods scale roughly linearly as $|\mathcal{V}|$ increase, but the querying time of TopChain increases more slowly when $|\mathcal{V}|$ becomes larger, i.e., from 4M to 8M. Figure~\ref{fig:power}(b) shows that the querying time does not change significantly with the increase in $\pi$. But the effect of $d_{\it avg}(u,\mathcal{G})$ is very different on the four methods, as shown in Figure~\ref{fig:power}(c). While the querying time of IP+ and GRAIL++ increases linearly as  $d_{\it avg}(u,\mathcal{G})$ increases, the query time of TopChain and Ferrari even decreases. This is because as $d_{\it avg}(u,\mathcal{G})$ increases, more vertices become reachable from others. TopChain and Ferrari can directly answer queries where the source query vertex can reach the target query vertex, while IP+ and GRAIL++ cannot avoid online search for such queries.

Overall, the results show that TopChain and Ferrari have better scalability than IP+ and GRAIL++, while TopChain is significantly faster than the other methods in all cases.


\section{Related Work}  \label{sec:related}

We first discuss related work on temporal graphs and then discuss existing reachability indexes.

\subsection{Related Work on Temporal Graphs} \label{ssec:related_temporal}
Studies on temporal graphs, also called \textbf{time-varying graphs}, are mainly related to temporal paths~\cite{CasteigtsFQS12paapp,KempeKK02jcss,KossinetsKW08kdd,Kostakos09physica,PanS11physrev,SantoroQFCA11corr,TangMML09sigcomm,TangMML10ccr,XuanFJ03ijfcs,WuCHKLX14pvldb} so far. Various types of temporal paths were defined to study temporal graphs~\cite{CasteigtsFQS12paapp,KossinetsKW08kdd,Kostakos09physica,PanS11physrev,SantoroQFCA11corr,TangMML09sigcomm,TangMML10ccr,XuanFJ03ijfcs,WuCHKLX14pvldb}. Temporal paths have been applied to study the connectivity of a temporal graph~\cite{KempeKK02jcss}, information latency in a temporal network~\cite{KossinetsKW08kdd}, small-world behavior~\cite{TangSMML10phyrev}, and for computing temporal connected components~\cite{NicosiaTMRML11corr,TangMML10ccr}. Temporal paths have also been used to define various metrics for temporal network analysis such as temporal efficiency and temporal clustering coefficient~\cite{TangMML09sigcomm,TangMML10ccr}, temporal betweenness~\cite{SantoroQFCA11corr} and closeness~\cite{PanS11physrev,SantoroQFCA11corr}. Most of the existing works were focused on concepts and measures for studying temporal graphs, while computational issues were largely ignored. Among these works, only~\cite{WuCHKLX14pvldb,XuanFJ03ijfcs} discussed algorithms for computing temporal paths, which are specific algorithms and are not designed for online query processing of temporal paths. TTL~\cite{WangLYXZ15sigmod} was proposed to answer online temporal path queries, but it is not scalable and its efficiency was only verified on small datasets. Apart from the above works, core decomposition in a large temporal graph was recently studied in~\cite{WuCLKHYW15bigdataconf}. Interesting readers can also refer to surveys on temporal graphs~\cite{CasteigtsFQS12paapp,HolmeS11corr,Muller-HannemannSWZ04atmos}.

\subsection{Related Work on Reachability Indexes} \label{ssec:related_reach}


Our method is also an indexing scheme for reachability querying. Existing reachability indexes can be mainly categorized into three groups: \emph{Transitive Closure}, \emph{2-Hop Labels}, \emph{Label+Search}. The transitive closure (TC) of a vertex $v$ is the set of vertices that $v$ can reach in $G$. Since the TCs are too large, existing methods in this category mainly attempt to reduce the TCs by various compression schemes~\cite{AgrawalBJ89sigmod,ChenC08icde,ChenC11icde,Jagadish90tods,JinRXW11tods,SchaikM11sigmod,WangHYYY06icde}. These methods are not scalable due to the high indexing cost. The 2-hop labeling scheme was first introduced in~\cite{CohenHKZ02soda}, which proved that computing a 2-hop label with minimum size is NP-hard, and proposed a $(\log |V|)$-approximation. Many following works have attempted to reduce the label size by various heuristics~\cite{BramandiaCN08www,CaiP10cikm,ChengYLWY08edbt,ChengHWF13sigmod,JinW13pvldb,SchenkelTW04edbt,YanoAIY13cikm,ZhuLWX14sigmod}, but all these methods are costly and cannot scale to large graphs. TTL~\cite{WangLYXZ15sigmod} is also a 2-hop indexing method, which is designed to answer queries of earliest-arrival time and the duration of a fastest path between two vertices within a time interval in a temporal graph. TTL cannot scale to large temporal graphs due to its expensive indexing cost, and its performance was only verified using small temporal graphs in~\cite{WangLYXZ15sigmod}. TTL also does not support dynamic update, which is necessary for temporal graphs.

The methods~\cite{ChenGK05vldb,SeufertABW13icde,TrisslL07sigmod,VelosoCJZ14edbt,WeiYLJ14pvldb,YildirimCZ12vldb} in the category of Label+Search construct a small index with a small construction cost, but their query performance is generally much worse than methods in the other two categories. However, the recent methods, \emph{Ferrari}~\cite{SeufertABW13icde} and \emph{IP+}~\cite{WeiYLJ14pvldb}, are able to achieve comparable query performance with methods in the other two categories. Ferrari applies tree cover to derive intervals so that reachability queries can be answered by checking the intervals~\cite{AgrawalBJ89sigmod}. However, the number of intervals is too large and Ferrari keeps only up to $k$ approximate intervals for every vertex. IP+ selects the top $k$ vertices, ranked based on independent permutation, that a vertex $v$ can reach or that can reach $v$ to be in $L_{out}(v)$ or $L_{in}(v)$. Thus, online search is needed to process some queries for both Ferrari and IP+.



TopChain is also a Label+Search approach, and the key idea of bounding the label size is similar to IP+ and Ferrari. However, our method is the only one that uses the properties of a temporal graph to design the labeling scheme. There are non-trivial issues in using these properties, and we provide detailed theoretical analysis to prove the correctness of labeling and querying. We also devise an efficient algorithm for dynamic update of the labels based on the graph properties, while both Ferrari and IP+ do not support update.

%
%

\section{Conclusions}  \label{sec:conclude}

In this paper, we presented TopChain, an efficient labeling scheme that employs the properties of a temporal graph for answering temporal reachability queries and time-based path queries. TopChain has a linear index construction time and linear index size, which makes the method scalable. TopChain significantly outperforms the state-of-the-art indexes~\cite{SeufertABW13icde,SchaikM11sigmod,WangLYXZ15sigmod,WeiYLJ14pvldb,YildirimCZ12vldb,ZhuLWX14sigmod}, and supports efficient dynamic update. As temporal graphs can be used to model many networks with time-ordered activities, TopChain is a useful tool for querying and analyzing these graphs. 


As for future work, we plan to apply TopChain to develop scalable systems for processing temporal graphs based on our existing work such as Husky~\cite{YangLC16pvldb}, Quegel~\cite{YanCOYLLZN16pvldb}, Blogel~\cite{YanCLN14pvldb} and Pregel+~\cite{YanCLN15www}.

\bibliographystyle{IEEEtran}
\bibliography{IEEEabrv,temreach}

\begin{thebibliography}{10}
\providecommand{\url}[1]{#1}
\csname url@samestyle\endcsname
\providecommand{\newblock}{\relax}
\providecommand{\bibinfo}[2]{#2}
\providecommand{\BIBentrySTDinterwordspacing}{\spaceskip=0pt\relax}
\providecommand{\BIBentryALTinterwordstretchfactor}{4}
\providecommand{\BIBentryALTinterwordspacing}{\spaceskip=\fontdimen2\font plus
\BIBentryALTinterwordstretchfactor\fontdimen3\font minus
  \fontdimen4\font\relax}
\providecommand{\BIBforeignlanguage}[2]{{%
\expandafter\ifx\csname l@#1\endcsname\relax
\typeout{** WARNING: IEEEtran.bst: No hyphenation pattern has been}%
\typeout{** loaded for the language `#1'. Using the pattern for}%
\typeout{** the default language instead.}%
\else
\language=\csname l@#1\endcsname
\fi
#2}}
\providecommand{\BIBdecl}{\relax}
\BIBdecl

\bibitem{KempeKK02jcss}
D.~Kempe, J.~M. Kleinberg, and A.~Kumar, ``Connectivity and inference problems
  for temporal networks,'' \emph{J. Comput. Syst. Sci.}, vol.~64, no.~4, pp.
  820--842, 2002.

\bibitem{PanS11physrev}
R.~K. Pan and J.~Saram{\"a}ki, ``Path lengths, correlations, and centrality in
  temporal networks,'' \emph{Phys. Rev. E}, vol.~84, p. 016105, 2011.

\bibitem{SantoroQFCA11corr}
N.~Santoro, W.~Quattrociocchi, P.~Flocchini, A.~Casteigts, and F.~Amblard,
  ``Time-varying graphs and social network analysis: Temporal indicators and
  metrics,'' \emph{CoRR}, vol. abs/1102.0629, 2011.

\bibitem{TangMML10ccr}
J.~Tang, M.~Musolesi, C.~Mascolo, and V.~Latora, ``Characterising temporal
  distance and reachability in mobile and online social networks,''
  \emph{Computer Communication Review}, vol.~40, no.~1, pp. 118--124, 2010.

\bibitem{ClementiP10}
A.~Clementi and F.~Pasquale, ``Information spreading in dynamic networks: An
  analytical approach,'' \emph{Theoretical Aspects of Distributed Computing in
  Sensor Networks}, 2010.

\bibitem{CasteigtsFMS11ipps}
A.~Casteigts, P.~Flocchini, B.~Mans, and N.~Santoro, ``Measuring temporal lags
  in delay-tolerant networks,'' in \emph{IPDPS}, 2011, pp. 209--218.

\bibitem{KossinetsKW08kdd}
G.~Kossinets, J.~M. Kleinberg, and D.~J. Watts, ``The structure of information
  pathways in a social communication network,'' in \emph{KDD}, 2008, pp.
  435--443.

\bibitem{TangSMML10phyrev}
J.~Tang, S.~Scellato, M.~Musolesi, C.~Mascolo, and V.~Latora, ``Small-world
  behavior in time-varying graphs,'' \emph{Physical Review E}, vol.~81, no.~5,
  p. 055101, 2010.

\bibitem{WuCHKLX14pvldb}
H.~Wu, J.~Cheng, S.~Huang, Y.~Ke, Y.~Lu, and Y.~Xu, ``Path problems in temporal
  graphs,'' \emph{PVLDB}, vol.~7, no.~9, pp. 721--732, 2014.

\bibitem{XuanFJ03ijfcs}
B.-M.~B. Xuan, A.~Ferreira, and A.~Jarry, ``Computing shortest, fastest, and
  foremost journeys in dynamic networks,'' \emph{Int. J. Found. Comput. Sci.},
  vol.~14, no.~2, pp. 267--285, 2003.

\bibitem{WangLYXZ15sigmod}
S.~Wang, W.~Lin, Y.~Yang, X.~Xiao, and S.~Zhou, ``Efficient route planning on
  public transportation networks: {A} labelling approach,'' in \emph{{SIGMOD}},
  2015, pp. 967--982.

\bibitem{ChengHWF13sigmod}
J.~Cheng, S.~Huang, H.~Wu, and A.~W. Fu, ``{TF-Label:} a topological-folding
  labeling scheme for reachability querying in a large graph,'' in
  \emph{{SIGMOD}}, 2013, pp. 193--204.

\bibitem{JinW13pvldb}
R.~Jin and G.~Wang, ``Simple, fast, and scalable reachability oracle,''
  \emph{{PVLDB}}, vol.~6, no.~14, pp. 1978--1989, 2013.

\bibitem{SeufertABW13icde}
S.~Seufert, A.~Anand, S.~J. Bedathur, and G.~Weikum, ``{FERRARI:} flexible and
  efficient reachability range assignment for graph indexing,'' in
  \emph{{ICDE}}, 2013, pp. 1009--1020.

\bibitem{TrisslL07sigmod}
S.~Tri{\ss}l and U.~Leser, ``Fast and practical indexing and querying of very
  large graphs,'' in \emph{{SIGMOD}}, 2007, pp. 845--856.

\bibitem{SchaikM11sigmod}
S.~J. van Schaik and O.~de~Moor, ``A memory efficient reachability data
  structure through bit vector compression,'' in \emph{{SIGMOD}}, 2011, pp.
  913--924.

\bibitem{WeiYLJ14pvldb}
H.~Wei, J.~X. Yu, C.~Lu, and R.~Jin, ``Reachability querying: An independent
  permutation labeling approach,'' \emph{{PVLDB}}, vol.~7, no.~12, pp.
  1191--1202, 2014.

\bibitem{YanoAIY13cikm}
Y.~Yano, T.~Akiba, Y.~Iwata, and Y.~Yoshida, ``Fast and scalable reachability
  queries on graphs by pruned labeling with landmarks and paths,'' in
  \emph{CIKM}, 2013.

\bibitem{YildirimCZ12vldb}
H.~Yildirim, V.~Chaoji, and M.~J. Zaki, ``{GRAIL:} a scalable index for
  reachability queries in very large graphs,'' \emph{{VLDB} J.}, vol.~21,
  no.~4, pp. 509--534, 2012.

\bibitem{ZhuLWX14sigmod}
A.~D. Zhu, W.~Lin, S.~Wang, and X.~Xiao, ``Reachability queries on large
  dynamic graphs: a total order approach,'' in \emph{SIGMOD}, 2014, pp.
  1323--1334.

\bibitem{CohenHKZ02soda}
E.~Cohen, E.~Halperin, H.~Kaplan, and U.~Zwick, ``Reachability and distance
  queries via 2-hop labels,'' in \emph{SODA}, 2002, pp. 937--946.

\bibitem{Simon88tcs}
K.~Simon, ``An improved algorithm for transitive closure on acyclic digraphs,''
  \emph{TCS}, vol.~58, pp. 325--346, 1988.

\bibitem{Jagadish90tods}
H.~V. Jagadish, ``A compression technique to materialize transitive closure,''
  \emph{TODS}, vol.~15, no.~4, pp. 558--598, 1990.

\bibitem{ChenC08icde}
Y.~Chen and Y.~Chen, ``An efficient algorithm for answering graph reachability
  queries,'' in \emph{ICDE}, 2008, pp. 893--902.

\bibitem{YildirimCZ13corr}
H.~Yildirim, V.~Chaoji, and M.~J. Zaki, ``{DAGGER:} {A} scalable index for
  reachability queries in large dynamic graphs,'' \emph{CoRR}, 2013.

\bibitem{YildirimCZ10pvldb}
------, ``{GRAIL:} scalable reachability index for large graphs,''
  \emph{{PVLDB}}, vol.~3, no.~1, pp. 276--284, 2010.

\bibitem{BramandiaCN10tkde}
R.~Bramandia, B.~Choi, and W.~K. Ng, ``Incremental maintenance of 2-hop
  labeling of large graphs,'' \emph{{IEEE} Trans. Knowl. Data Eng.}, vol.~22,
  no.~5, pp. 682--698, 2010.

\bibitem{DemetrescuI06jcss}
C.~Demetrescu and G.~F. Italiano, ``Fully dynamic all pairs shortest paths with
  real edge weights,'' \emph{J. Comput. Syst. Sci.}, vol.~72, no.~5, pp.
  813--837, 2006.

\bibitem{RodittyZ04stoc}
L.~Roditty and U.~Zwick, ``A fully dynamic reachability algorithm for directed
  graphs with an almost linear update time,'' in \emph{STOC}, 2004, pp.
  184--191.

\bibitem{SchenkelTW05icde}
R.~Schenkel, A.~Theobald, and G.~Weikum, ``Efficient creation and incremental
  maintenance of the {HOPI} index for complex {XML} document collections,'' in
  \emph{{ICDE}}, 2005, pp. 360--371.

\bibitem{CasteigtsFQS12paapp}
A.~Casteigts, P.~Flocchini, W.~Quattrociocchi, and N.~Santoro, ``Time-varying
  graphs and dynamic networks,'' \emph{International Journal of Parallel,
  Emergent and Distributed Systems}, vol.~27, no.~5, pp. 387--408, 2012.

\bibitem{Kostakos09physica}
V.~Kostakos, ``Temporal graphs,'' \emph{Physica A: Statistical Mechanics and
  its Applications}, vol. 388, no.~6, pp. 1007--1023, 2009.

\bibitem{TangMML09sigcomm}
J.~Tang, M.~Musolesi, C.~Mascolo, and V.~Latora, ``Temporal distance metrics
  for social network analysis,'' in \emph{WOSN}, 2009, pp. 31--36.

\bibitem{NicosiaTMRML11corr}
V.~Nicosia, J.~Tang, M.~Musolesi, G.~Russo, C.~Mascolo, and V.~Latora,
  ``Components in time-varying graphs,'' \emph{CoRR}, vol. abs/1106.2134, 2011.

\bibitem{WuCLKHYW15bigdataconf}
H.~Wu, J.~Cheng, Y.~Lu, Y.~Ke, Y.~Huang, D.~Yan, and H.~Wu, ``Core
  decomposition in large temporal graphs,'' in \emph{{IEEE} International
  Conference on Big Data}, 2015, pp. 649--658.

\bibitem{HolmeS11corr}
P.~Holme and J.~Saram{\"a}ki, ``Temporal networks,'' \emph{CoRR}, vol.
  abs/1108.1780, 2011.

\bibitem{Muller-HannemannSWZ04atmos}
M.~M{\"{u}}ller{-}Hannemann, F.~Schulz, D.~Wagner, and C.~D. Zaroliagis,
  ``Timetable information: Models and algorithms,'' in \emph{{ATMOS}}, 2004,
  pp. 67--90.

\bibitem{AgrawalBJ89sigmod}
R.~Agrawal, A.~Borgida, and H.~V. Jagadish, ``Efficient management of
  transitive relationships in large data and knowledge bases,'' in
  \emph{{SIGMOD}}, 1989, pp. 253--262.

\bibitem{ChenC11icde}
Y.~Chen and Y.~Chen, ``Decomposing {DAG}s into spanning trees: A new way to
  compress transitive closures,'' in \emph{ICDE}, 2011, pp. 1007--1018.

\bibitem{JinRXW11tods}
R.~Jin, N.~Ruan, Y.~Xiang, and H.~Wang, ``Path-tree: An efficient reachability
  indexing scheme for large directed graphs,'' \emph{ACM Trans. Database
  Syst.}, vol.~36, no.~1, p.~7, 2011.

\bibitem{WangHYYY06icde}
H.~Wang, H.~He, J.~Yang, P.~S. Yu, and J.~X. Yu, ``Dual labeling: Answering
  graph reachability queries in constant time,'' in \emph{{ICDE}}, 2006, p.~75.

\bibitem{BramandiaCN08www}
R.~Bramandia, B.~Choi, and W.~K. Ng, ``On incremental maintenance of 2-hop
  labeling of graphs,'' in \emph{WWW}, 2008, pp. 845--854.

\bibitem{CaiP10cikm}
J.~Cai and C.~K. Poon, ``Path-hop: efficiently indexing large graphs for
  reachability queries,'' in \emph{CIKM}, 2010, pp. 119--128.

\bibitem{ChengYLWY08edbt}
J.~Cheng, J.~X. Yu, X.~Lin, H.~Wang, and P.~S. Yu, ``Fast computing
  reachability labelings for large graphs with high compression rate,'' in
  \emph{{EDBT}}, 2008, pp. 193--204.

\bibitem{SchenkelTW04edbt}
R.~Schenkel, A.~Theobald, and G.~Weikum, ``{HOPI:} an efficient connection
  index for complex {XML} document collections,'' in \emph{{EDBT}}, 2004, pp.
  237--255.

\bibitem{ChenGK05vldb}
L.~Chen, A.~Gupta, and M.~E. Kurul, ``Stack-based algorithms for pattern
  matching on dags,'' in \emph{VLDB}, 2005, pp. 493--504.

\bibitem{VelosoCJZ14edbt}
R.~R. Veloso, L.~Cerf, W.~M. Junior, and M.~J. Zaki, ``Reachability queries in
  very large graphs: {A} fast refined online search approach,'' in \emph{EDBT},
  2014, pp. 511--522.

\bibitem{YangLC16pvldb}
F.~Yang, J.~Li, and J.~Cheng, ``Husky: Towards a more efficient and expressive
  distributed computing framework,'' \emph{{PVLDB}}, vol.~9, no.~5, pp.
  420--431, 2016.

\bibitem{YanCOYLLZN16pvldb}
D.~Yan, J.~Cheng, T.~Ozsu, F.~Yang, Y.~Lu, J.~C.~S. Lui, Q.~Zhang, and W.~Ng,
  ``A general-purpose query-centric framework for querying big graphs
  [innovative systems and applications],'' \emph{{PVLDB}}, vol.~9, no.~7, 2016.

\bibitem{YanCLN14pvldb}
D.~Yan, J.~Cheng, Y.~Lu, and W.~Ng, ``Blogel: {A} block-centric framework for
  distributed computation on real-world graphs,'' \emph{{PVLDB}}, vol.~7,
  no.~14, pp. 1981--1992, 2014.

\bibitem{YanCLN15www}
------, ``Effective techniques for message reduction and load balancing in
  distributed graph computation,'' in \emph{{WWW}}, 2015, pp. 1307--1317.

\end{thebibliography}

\end{document}